\documentclass[11pt]{article}

\makeatletter
\renewcommand\section{\@startsection {section}{1}{\z@}%
                                   {-3.5ex \@plus -1ex \@minus -.2ex}%
                                   {2.3ex \@plus.2ex}%
                                   {\normalfont\bfseries}}

\def\@maketitle{%
  \newpage
  \null
  \vskip 2em%
  \begin{center}%
  \let \footnote \thanks
    {\large \@title \par}%
    \vskip 1.5em%
    {\normalsize
      \lineskip .5em%
      \begin{tabular}[t]{c}%
        \@author
      \end{tabular}\par}%
    \vskip 1em%
  \end{center}%
  \par
  \vskip 1.5em}
\makeatother
\usepackage{graphicx}

\usepackage{amsmath,amsfonts,amscd,bezier,amssymb}
\usepackage{amssymb}
\usepackage[normalem]{ulem}
\usepackage[all]{xy}

\usepackage[latin1]{inputenc}

\usepackage{color}

\newenvironment{proof}{\noindent\textbf{Proof: }}{\hfill \small $\Box$}

\newtheorem{theorem}{Theorem}
\newtheorem{lemma}[theorem]{Lemma}
\newtheorem{proposition}[theorem]{Proposition}
\newtheorem{corollary}[theorem]{Corollary}
\newtheorem{definition}[theorem]{Definition}
 \newtheorem{example}{Example}
 \newtheorem{remark}{Remark}

\begin{document}

\title{\bf $n$-Dimensional Fuzzy Negations}


\author{{\bf Benjamin Bedregal$^a$ and Ivan Mezzomo$^b$ and Renata H.S. Reiser$^c$ } \\
\small $^a$ Department of Informatics and Applied Mathematics, \\ 
\small Federal University of Rio Grande do Norte \\
\small \textit{bedregal@dimap.ufrn.br} \\
\small $^b$ Center of Exact and Natural Sciences,   \\ 
\small Rural Federal University of SemiArid \\
\small \textit{imezzomo@ufersa.edu.br} \\
\small $^c$ Center of Tecnological Development, \\
\small  Federal University of Pelotas\\ 
\small		\textit{reiser@inf.ufpel.edu.br} 
}

\maketitle

\begin{abstract} 
$n$-Dimensional fuzzy sets is a fuzzy set extension where the membership values are $n$-tuples of real numbers in the unit interval $[0,1]$ orderly increased, called $n$-dimensional intervals. The set of $n$-dimensional intervals is denoted by $L_n([0,1])$. This paper aims to investigate a special extension from $[0,1]$ -- $n$-representable fuzzy negations on  $L_n([0,1])$, summarizing the class of such functions which are continuous and monotone by part.
The main properties of (strong) fuzzy negations on $[0,1]$ are preserved by representable (strong) fuzzy negation on $L_n([0,1])$, mainly related to the analysis of degenerate elements and equilibrium points. The conjugate obtained by action of an $n$-dimensional automorphism on an $n$-dimensional fuzzy negation  provides a method to obtain other $n$-dimensional fuzzy negation, in which properties such as representability, continuity and monotonicity  on $L_n([0,1])$ are preserved.
\end{abstract}
\vspace{12pt}\noindent{\bf Keywords} Fuzzy negations, $n$-dimensional fuzzy sets, $n$-dimensional fuzzy negations, $n$-dimensional automorphisms.

\section{Introduction}

The notion of an $n$-dimensional fuzzy set on $L_n$-fuzzy set theory was introduced in~\cite{shang} as a special class of $L$-fuzzy set theory,  generalizing the theories underlying  the fuzzy logic and many other multivalued fuzzy logics: the interval-valued fuzzy set, the intuitionistic fuzzy set, the interval-valued intuitionistic fuzzy set and the type-2 fuzzy logic~\cite{bustince15}.

In accordance with the  Zadeh's Extension Principle,  $L_n$-fuzzy set theory provides  additional degrees of freedom that makes it possible to directly model uncertainties in computational systems based on fuzzy logics. Such uncertainties  are frequently associated to systems where time-varying,  non stationary statistical attributes  or knowledge of experts using questionnaires including uncertain words from natural language. However, the corresponding mathematical description of such models is unknown or not totally consolidated yet.

This paper considers the main properties of an $n$-dimensional fuzzy set $A$  over a reference set $X$, where each  element $x \in X \neq \emptyset$ is related with 
an $n$-dimensional interval, characterized by its ordered $n$-membership values: $\mu_{A_1}(x)\leq \ldots \leq \mu_{A_n}(x)$. Thus,  for  $i=1, \ldots, n$, each $n$-membership function  $\mu_{A_i}:X\rightarrow [0,1]$, called as the $i$-th membership  degree of $A$, can provide an interpretation to model the uncertainty of $n$-distinct parameters from  evaluation processes or fuzzy measures in computational systems modelled by $L_n$-fuzzy set theory.

\subsection{Main contribution}

The main contribution  of this paper is concerned with representability of fuzzy negations on the set of $n$-dimensional intervals, denoted by $L_n([0,1])$, specially related to $\subseteq$-monotonicity and monotonicity by part of corresponding $n$-membership function. These topics are closely connected to degenerate elements and equilibrium points of $n$-dimensional fuzzy negations.  

By considering an $n$-dimensional fuzzy negation  $\mathcal{N}$ and related $n$-projections, the $n$-representability of $\mathcal{N}$ is discussed and the notion of $\subseteq_i$-monotonicity is formalized  on $L_n([0,1])$, for $i=0, \ldots , n$. 

Additionally, it is shown that the partial order of fuzzy negation  can be extended from $[0,1]$ to $L_n([0,1])$. 

Moreover, several propositions were offered on possible conditions under which main properties of strong fuzzy negations on $[0,1]$ are preserved by representable strong $n$-dimensional fuzzy negations on $L_n([0,1])$. In particular, these propositions also  guarantee  such $n$-dimensional fuzzy negations as operators preserving degenerate elements on $L_n([0,1])$.

Concepts and intrinsic properties of conjugated fuzzy negations obtained by action of  automorphisms on $[0,1]$  have a counterpart on $L_n([0,1])$.
The conjugate notion of $n$-dimensional fuzzy negations, which can be generated by action of $n$-dimensional automorphisms is studied. The paper also investigates the conditions under which equilibrium points and degenerate elements are preserved by conjugate fuzzy negations on $L_n([0,1])$. 

\subsection{Main related papers}
In~\cite{shang}, the definitions of cut set on an n-dimensional fuzzy set and its corresponding  n-dimensional vector level cut set of Zadeh fuzzy set are presented in order to study not only decomposition but also  representation theorems of the $n$-dimensional fuzzy sets. Thus, new decomposition and representation theorems of the Zadeh fuzzy set are proposed.

In \cite{bedregal12}, the authors consider the study of aggregation operators for these new concepts of $n$-dimensional fuzzy sets, starting from the usual aggregation operator theory and also including a new class of aggregation operators containing an extension of the OWA operator, which is  based on $n$-dimensional fuzzy connectives. The results presented in such context allow to extend fuzzy sets to interval-valued Atanassov's intuitionistic fuzzy sets and also preserve their main properties. In particular, in \cite{bedregal12}, it was introduced the notion of $n$-dimensional fuzzy negation and showed a way of building $n$-dimensional fuzzy negation from $n$ comparable fuzzy negations. In \cite{mezzomo16}, we were introduced and studied the $n$-dimensional strict fuzzy negations. 

In the context of lattice-valued fuzzy set theory \cite{goguen}, the notion of fuzzy connectives for lattice-valued fuzzy logics was generalized in~\cite{baets, IPMU12, helida, palmeira} by taking into account axiomatic definitions.  
In \cite{palmeira14}, it was extended the notion $n$-dimensional fuzzy set by considered arbitrary bounded lattice $L$ and, in\cite{palmeira}, it was introduced the notion of $n$-dimensional lattice-valued negation.

Following the results in the above cited works, this paper studies the possibility of dealing with main properties of representable fuzzy negation on $L_n([0,1])$ and  obtaining other ones by action of $n$-dimensional automorphisms. In  particular, we studied the $n$-dimensional strong fuzzy negations.


\section{Preliminaries}\label{sect2}

In this section, we will briefly review some basic concepts  which are necessary for the development of this paper. The previous main definitions and additional results concerning the study of $n$-dimensional fuzzy negations presented in this work can be found in \cite{alsina, atanassov99, bedregal10, bedregal12, beliakov,  bustince03, bustince99, costa, FodorRoubens1994, klement}.

\subsection{Automorphisms}

According with~\cite[Definition 0]{bustince03},  a  function $\rho\colon [0,1]\to [0,1]$ is an \textbf{automorphism} if it is continuous, strictly increasing and verifies the boundary conditions $\rho(0)=0$ and $\rho(1)=1$, i.e., if it is an increasing bijection on $U$, meaning that for each $x,y\in [0,1]$, if $x\leq y$, then $\rho(x)\leq\rho(y)$.

Automorphisms are closed under composition, i.e., denoting $\mathcal{A}([0,1])$ the set of all automorphisms on $[0,1]$,
if $\rho, \rho' \in \mathcal{A}([0,1])$ then $\rho\circ\rho'(x)=\rho(\rho'(x)) \in \mathcal{A}([0,1])$. In addition, the inverse $\rho^{-1}$ of an automorphism $\rho$ is 
also an automorphism, meaning that $\rho^{-1}(x) \in \mathcal{A}([0,1])$.

By~\cite{{bustince03}}, the action of an automorphism $\rho$ on a function $f\colon[0,1]^{n}\to[0,1]$, denoted by $f^{\rho}$ and named \textbf{the $\rho$-conjugate of $f$} is defined as, for all $(x_1, \ldots , x_n){\in} [0,1]^{n}$:
\begin{equation}
\label{eq:aut-imp}
  f^\rho(x_1, \ldots , x_n){=} \rho^{-1}(f (\rho(x_1),\ldots ,\rho(x_n))).
\end{equation}
 
Let $f_1, f_2\colon [0,1]^n \to [0,1]$ be functions. The functions $f_1$ and $f_2$ are \textbf{conjugated to each other}, if there exists an automorphism $\rho$ such that
\begin{equation*}
  f_2(x_1, \ldots , x_n) {=} \rho^{-1}(f_1 (\rho(x_1),\ldots ,\rho(x_n))),  
\end{equation*}
for all $(x_1, \ldots , x_n){\in} [0,1]^{n}$. Notice that, if $f_2=f_{1}^{\rho}$ then $f_{1} = f_{2}^{\rho^{-1}}$.

\subsection{Aggregations}
Let $n\in \mathbb{N}$ such that $n\geq 2$. A function $A:[0,1]^{n}\rightarrow [0,1]$ is an $n$-ary \textbf{aggregation operator} if, for each $x_1, \ldots, x_n,y_1, \ldots,y_n\in [0,1]$, $A$ satisfies the following conditions:

 A1. $A(0, \ldots, 0)=0$ and  $A(1, \ldots, 1)=1$;
 
 A2.  If $x_i\leq y_i$, for each $i=1,\ldots, n$, then $A(x_1,\ldots,x_n)\leq A(y_1, \ldots, y_n)$.

\subsection{Fuzzy negations}
A function $N:[0,1]\rightarrow [0,1]$ is a \textbf{fuzzy negation} if

 N1: $N(0)=1$ and $N(1)=0$;

 N2: If $x\leq y$, then $N(x)\geq N(y)$, for all $x,y\in [0,1]$.\\
A fuzzy negation $N$ satisfying the involutive property

 N3: $N(N(x))=x$, for all $x\in [0,1]$,\\
  is called \textbf{strong fuzzy negation}. And, a continuous fuzzy negation $N$ is strict if it verifies

 N4: $N(x)<N(y)$ when $y<x$, for all $x\in [0,1]$. 

Strong fuzzy negations are also strict fuzzy negations~\cite{klement}. The standard strong fuzzy negation is defined as $N_S(x) = 1 -x$.


An \textbf{equilibrium point} of a fuzzy negation $N$ is a value $e\in [0,1]$ such that $N(e) = e$. See~\cite[Remarks 2.1 and 2.2]{bedregal10} and~\cite[Proposition 2.1]{bedregal10} for additional studies related to main properties of equilibrium points. 

%
%
%
%
%

\begin{example}\label{ex:1}
The function $C_k:[0,1] \rightarrow [0,1]$  given by 
\begin{eqnarray}\label{eq_ckn}
C_k(x)= \sqrt[n-k+1]{1-x^{n-k+1}}, \, \forall ~k \in \{1,2, \ldots , n\}
\end{eqnarray}
 is a strong fuzzy negation. Since $C_k$ is strong,  it is also  a strict fuzzy negation. Moreover, based on \cite[Theorem 3.4]{klir95}, every continuous fuzzy negation has a unique equilibrium point. So, $C^k$ has a unique equilibrium point. Notice that $e=\sqrt[n-k+1]{\frac{1}{2}}$ is the equilibrium point of $C_k$, i.e., $\sqrt[n-k+1]{1-e^{n-k+1}}=e$.
\hfill\rule{2mm}{2mm} 
\end{example}


%
%




\begin{proposition}\label{pro-dk}
The function $C^k:[0,1] \rightarrow [0,1]$  given by 
\begin{eqnarray}\label{eq-dk}
C^k(x) &=& 1-x^{k}, \,\,\,  \forall~ k \in \{2, 3, \ldots , n\}.
\end{eqnarray}
 is a strict but not strong fuzzy negation.
\end{proposition}
\begin{proof}
Straightforward from \cite{bedregal10}.
\end{proof}\\

Fuzzy negations have at most one equilibrium point, as proved by  Klir and Yuan, in \cite[Theorem 3.2]{klir95}. Therefore, if a fuzzy negation has an equilibrium point then it is unique. However, not all fuzzy negations have an equilibrium point \cite{bedregal10}, e.g. the fuzzy negations $N_\bot$ and $N_\top$, respectively given as:
\begin{eqnarray*}
N_\bot(x)=\left\{
\begin{array}{ccl}
0, & \mbox{if} & x>0; \\
1, & \mbox{if} & x= 0;
\end{array}
\right. &&
N_\top(x)=\left\{
\begin{array}{ccl}
0, & \mbox{if} & x=1; \\
1, & \mbox{if} & x< 1. 
\end{array}
\right.
\end{eqnarray*}

Clearly, for all fuzzy negation $N$, it holds that $N_\bot\leq N\leq N_\top$.


In \cite[Prop.~4.2]{navara99}, Navara introduced the notion of negation-preserving automorphisms, assuring that an \textbf{$N_S$-preserving automorphism} $\rho  \in Aut([0,1])$ 
 commutes with the usual fuzzy negation $N_S$, meaning that $\rho(N_S(x))= N_S(\rho(x))$, for all $x \in [0,1]$. Additionally, a natural generalization of such notion is given by Bedregal, in~\cite{bedregal10}:  Let $N$ be a fuzzy negation.  A function $\rho \in Aut([0,1])$ is an \textbf{$N$-preserving automorphism}  if and only if $\rho$ verifies the condition
\begin{eqnarray}\label{preserving auto}
\rho(N(x)) & = & N(\rho(x)),  \qquad  \forall ~x\in [0,1].
\end{eqnarray}

 The Navara's characterization for negation-preserving automorphisms is generalized below.

\begin{proposition}\cite[Proposition~2.6]{bedregal10}\label{bedregalprop2.6}
Let $N$ be a strong fuzzy negation which has $e$ as the unique equilibrium point of $N$. When $\rho \in Aut([0,e])$ then $\rho^N: [0,1]\rightarrow [0,1]$, defined by 
\begin{eqnarray}\label{rhoN}
\rho^N(x)=\left\{
\begin{array}{lll}
\rho(x), & \mbox{if} & x\leq e; \\
(N \circ \rho \circ N) (x), & \mbox{if} & x> e 
\end{array}
\right.
\end{eqnarray}
is an $N$-preserving automorphism. 
\end{proposition}

Additionally,  $N$-preserving automorphisms are given as Eq.(\ref{rhoN}).

\begin{proposition}\cite[Proposition~2.7]{bedregal10}\label{bedregalprop2.7}
Let $N$ be a strong fuzzy negation which has $e$ as the unique equilibrium point. When $\rho \in Aut([0,e])$ then $\rho^{N^{ -1}}$ is an $N$-preserving automorphism. 
\end{proposition}

\subsection{$n$-Dimensional fuzzy sets}

Let $X$ be a non empty set and $n\in\mathbb{N}^{+}=\mathbb{N}-\{0\}$. According to \cite{shang},\textbf{ an $n$-dimensional fuzzy set $A$ over $X$} is given by 
\begin{center}
$A=\{(x, \mu_{A_1}(x), \ldots, \mu_{A_n}(x)): x\in X\}$,
\end{center}
 where, for each $i=1, \ldots, n$, $\mu_{A_i}:X\rightarrow [0,1]$  is called $i$-th membership  degree of $A$, which also satisfies the condition: $\mu_{A_1}(x)\leq \ldots \leq \mu_{A_n}(x)$, for $x \in X$. 

In \cite{bedregal11}, for  $n\geq 1$, an \textbf{$n$-dimensional upper simplex} is given as 
\begin{eqnarray}
L_n([0,1])=\{(x_1, \ldots, x_n)\in [0,1]^{n}:x_1\leq \ldots \leq x_n\},
\label{L_n}
\end{eqnarray}
and its elements are called \textbf{$n$-dimensional intervals}. 

For each $i=1, \ldots, n$, the $i$-th projection of $L_n([0,1])$ is the function $\pi_i:L_n([0,1])\rightarrow [0,1]$ defined by $\pi_i(x_1, \ldots , x_n)=x_i$. 

Notice that $L_1([0,1])=[0,1]$ and $L_2([0,1])$ reduces to the usual lattice  of all the closed subintervals of the unit interval $[0,1]$. 

A \textbf{degenerate element}  $\textbf{x}\in L_n([0,1])$ verifies the following condition  
\begin{equation}\pi_i(\textbf{x})=\pi_j(\textbf{x}), \,\,\, \forall~ i,j=1, \ldots, n.
\end{equation} 
The degenerate element $(x, \ldots,x)$ of $L_n([0,1])$, for each $x\in [0,1]$, will be denoted by $/x/$ and the set of all degenerate elements of $L_n([0,1])$ will be denoted by $\mathcal{D}_n$. 

An $m$-ary function $F:L_n([0,1])^{m}\rightarrow L_n([0,1])$ is called \textbf{$\mathcal{D}_n$-preserve} function or a function  preserving degenerate elements if  the following condition holds

\noindent(\textbf{DP})\hspace{1cm} $F(\mathcal{D}^{m}_n) = F(/x_1/, \ldots, /x_m/)\in\mathcal{D}_n, \\ \forall x_1, \ldots, x_m\in [0,1]$.

Based on~\cite{bedregal11}, the supremum and infimum  on $L_n([0,1])$ are both given, for all $\textbf{x},\textbf{y} \in L_n([0,1])$, as
\begin{eqnarray}\label{vee}
 \textbf{x}\vee \textbf{y}&=&(\max(x_1,y_1), \ldots, \max(x_n,y_n)),\\
\label{wedge}
 \textbf{x}\wedge \textbf{y}&=&(\min(x_1,y_1), \ldots, \min(x_n,y_n)).
\end{eqnarray}
And, by considering the natural extension of the order $\leq$ on $L_2([0,1])$ as in \cite{bedregal10,bustince99} to higher dimensions,   for all $\textbf{x},\textbf{y} \in L_n([0,1])$, it holds that
\begin{eqnarray}\label{mono_proji}
\textbf{x}\leq\textbf{y}~\mbox{iff} ~ \pi_i(\textbf{x})\leq \pi_i(\textbf{y}), \,\,\  ~\forall ~i=1, \ldots, n.
\label{xleqy}
\end{eqnarray}

 \subsection{Continuous $n$-dimensional interval function}

 In the following, the continuity of a function $F: L_n([0,1])\rightarrow L_n([0,1])$, called an  $n$-dimensional function or an  
 \textbf{$n$-dimensional interval function}, will be studied, based on the continuity on $L([0,1]^n)$.

\begin{proposition}\label{PropCont}
Let $\rho:[0,1]^n\to L_n([0,1])$ be the function defined by $\rho(x_1, \ldots, x_n)=[x_{(1)}, \ldots, x_{(n)}]$, when $(x_{(1)}, \ldots, x_{(n)})$ is a fixed-permutation of a tuple $(x_1, \ldots, x_n)$ such that $x_{(i)}\leq x_{(i+1)}$ for all $i=1, \ldots, n-1$ and let $\sigma:L_n([0,1])\rightarrow [0,1]^n$ be the function defined by $\sigma([x_1, \ldots, x_n])= (x_1, \ldots, x_n)$. For an $n$-dimensional interval function  $F: L_n([0,1])\rightarrow L_n([0,1])$, the corresponding operator $F^{\rho}: [0,1]^n \rightarrow [0,1]^n$ given by
\begin{equation}\label{eq-f-rho}
F^\rho(x_1, \ldots, x_n)=\sigma( F(\rho(x_1, \ldots, x_n))), 
\end{equation}
 for all $(x_1, \ldots, x_n) \in  L_n([0,1])$, is a non-injective function. 
\end{proposition}
\begin{proof}
Straightforward. 
\end{proof}

Based on the above results, from Eq.~(\ref{vee}) to Eq.~(\ref{mono_proji}), the lattice \mbox{$(L_n([0,1]),\leq)$}  satisfies the \textbf{notion of continuity} given as follows:

\begin{definition}\label{DefCont}
An $n$-dimensional function $F: L_n([0,1])\rightarrow L_n([0,1])$ is continuous if the function $F^\rho:[0,1]^n\rightarrow [0,1]^n$ given by Eq.~(\ref{eq-f-rho}) is also continuous. 
\end{definition}

Hence, the continuity of an $n$-dimensional function $F: L_n([0,1])\rightarrow L_n([0,1])$  is induced by the composition $F^\rho = \sigma \circ  F \circ \rho$  
from the usual continuity notion on $[0,1]^n$.

\section{Fuzzy negations on $L_n([0,1])$}\label{sec-3}

In this section, we study the notion of fuzzy negation on the lattice $(L_n([0,1]), \leq)$ as conceived by Bedregal in \cite{bedregal12} 
and their relation with usual notion of fuzzy negation. 

\begin{definition}
A function  $\mathcal{N}:L_n([0,1])\rightarrow L_n([0,1])$  is an \textbf{$n$-dimensional fuzzy negation} if it satisfies the following properties:

\textbf{N1}: $\mathcal{N}(/0/)=/1/$ and $\mathcal{N}(/1/)=/0/$;

\textbf{N2}:  If ${\bf x}\leq {\bf y}$ then $\mathcal{N}({\bf x})\geq\mathcal{N}({\bf y})$, for all ${\bf x}, {\bf y} \in L_n([0,1])$.
\end{definition}

\begin{proposition}\cite[Proposition 3.1]{bedregal12}\label{tilden}
Let $N_1, \ldots, N_n$ be fuzzy negations such that $N_1\leq\ldots\leq N_n$. Then $\widetilde{N_1 \ldots N_n}: L_n([0,1])\rightarrow L_n([0,1])$ defined by
\begin{eqnarray}
\widetilde{N_1\ldots N_n}({\bf x})=(N_1(\pi_n({\bf x})), \ldots, N_n(\pi_1({\bf x})))
\label{tildeN}
\end{eqnarray}
is an $n$-dimensional fuzzy negation.
\end{proposition}

Additionally, according to \cite{bedregal11}, when $i=1, \ldots, n-1$, the $\subseteq_i$-relation with respect to the $i$-th component of $\textbf{x}, \textbf{y} \in L_n([0,1])$ is given as the following 
\begin{equation}\label{eq_pi}
\textbf{x}\subseteq_i \textbf{y} \mbox{ when } \pi_i(\textbf{y})\leq \pi_i(\textbf{x})\leq \pi_{i+1}(\textbf{x})\leq\pi_{i+1}(\textbf{y}).
\end{equation}

\subsection{Representability and monotonicity of fuzzy negations on $L_n([0,1])$}

In order to analyse  properties related to equilibrium point,  representable and monotone fuzzy negations on $L_n([0,1])$ are firstly studied in this section.

An $n$-dimensional fuzzy negation $\mathcal{N}$ is called \textbf{$n$-representable} if 
there exist fuzzy negations $N_1,\ldots ,N_n$ such that 
\begin{equation}\label{ineq_neg}
N_1\leq\ldots\leq N_n \mbox{ and } \mathcal{N}=\widetilde{N_1 \ldots N_n}.
\end{equation} 
By reducing notation, when $N_i=N$ for all $i=1, \ldots, n$,  an $n$-representable fuzzy negation $\widetilde{N \ldots N}$ will be denoted by $\widetilde{N}$.

\begin{proposition} Let $C^k(x) = 1-x^{k}$, for $k \in \{1,2, \ldots , n\}$. Then an $n$-dimensional function $\mathcal{C}: L_n([0,1])\rightarrow L_n([0,1])$ given as the 
following
\begin{equation}
\mathcal{C}({\bf{x}}) = \widetilde{C^1\ldots C^n}({\bf x})=(C^1(\pi_n({\bf x})), \ldots, C^n(\pi_1({\bf x})))
\end{equation} 
is an $n$-representable fuzzy negation on $L_n([0,1])$.
\end{proposition}

\begin{proof}
For all $1 \leq  i \leq j  \leq n$ and $x\in [0,1]$, we have that   $1-x^{i} \leq 1-x^j$, resulting the following inequalities:  $C^{1} \leq  C^{i} \leq C^{j} \leq C^{n}$. 
Additionally, the following is verified:

 \textbf{N1:}  $\widetilde{C^1{\ldots} C^n}(/0/)=  (C^1(\pi_n(/0/)), {\ldots}, C^n(\pi_1(/0/))) =  (C^1(0), {\ldots}, C^n(0))=  /1/$; and \\
$\widetilde{C^1{\ldots} C^n}(/1/)=(C^1(\pi_n(/1/)), {\ldots}, C^n(\pi_1(/1/)))= (C^1(1), {\ldots}, C^n(1))=  /0/$;

\textbf{N2:} Based on monotonicity of projection-functions,  if $\bf{x}\geq \bf{y}$ then $$(C^1(\pi_n({\bf x})), \ldots, C^n(\pi_1({\bf x}))) \leq (C^1(\pi_n({\bf y})), \ldots, C^n(\pi_1({\bf y}))),$$therefore, $\widetilde{C^1\ldots C^n}({\bf x}) \leq \widetilde{C^1\ldots C^n}({\bf y})$ meaning that $\mathcal{C}(\bf{x}) \leq \mathcal{C}(\bf{y})$.

Concluding, $\mathcal{C} = \widetilde{C^1{\ldots} C^n}$ is an $n$-representable fuzzy negation on $L_n([0,1])$. 
\end{proof}\\

 Let $i\in \{1,\ldots,n-1\}$. An $n$-dimensional fuzzy negation $\mathcal{N}$ is called \textbf{$\subseteq_i$-monotone} if, for any $\textbf{x}, \textbf{y}\in L_n([0,1])$ 
 it holds that 
 \begin{equation}\label{cont_mono_i}
 \mathcal{N}(\textbf{x})\subseteq_i\mathcal{N}(\textbf{y}) \mbox{ whenever } \textbf{x}\subseteq_{n-i} \textbf{y}.
 \end{equation}
 
  Moreover, one can  say that an $n$-dimensional fuzzy negation $\mathcal{N}$ is called

 (i)  \textbf{$\subseteq$-monotone} if $\mathcal{N}$ is $\subseteq_i$-monotone for all $i=1, \ldots, n-1$; and
 
 (ii)  \textbf{monotone by part} when, for all $i=1, \ldots, n$ and $\textbf{x}, \textbf{y}\in L_n([0,1])$, 
\begin{equation}\label{eq_pi_ni}
\pi_{i}(\mathcal{N}(\textbf{x}))\leq\pi_{i}(\mathcal{N}(\textbf{y})) \mbox{ whenever } \pi_{n-i+1}(\textbf{x})\geq \pi_{n-i+1}(\textbf{y}).
\end{equation}

\begin{remark}
	When we say that $\mathcal{N}$ is $\subseteq_i$-monotone, for all $i=1,\ldots, n-1$, it does not mean that we will only consider ${\bf x}, {\bf y} \in L_n([0,1])$ such that $\pi_i({\bf y})\leq \pi_i({\bf x})\leq \pi_{i+1}({\bf x})\leq \pi_{i+1}({\bf y})$, for all $i=1, \ldots, n-1$. Instead, we consider all ${\bf x}, {\bf y} \in L_n([0,1])$ and if for some $i$,  ${\bf x}\subseteq_i {\bf y}$ then by Eq. (\ref{cont_mono_i}) we have that $\mathcal{N}({\bf x})\subseteq_{n-i} \mathcal{N}({\bf y})$. For example, consider a $n$-dimensional fuzzy negation $\mathcal{N}$ which is $\subseteq$-monotone, ${\bf x}=(0.2, 0.4, 0.7, 0.8, 0,9)$ and ${\bf y}=(0.1, 0.5, 0.6, 0.8, 1)$. Clearly, ${\bf x}\subseteq_1 {\bf y}$, ${\bf x}\subseteq_3 {\bf y}$ and ${\bf x}\subseteq_4 {\bf y}$ but  ${\bf x}\not\subseteq_2 {\bf y}$.  Since $\mathcal{N}$ is $\subseteq$-monotone we can conclude that of $\mathcal{N}({\bf x})\subseteq_4 \mathcal{N}({\bf y})$, $\mathcal{N}({\bf x})\subseteq_2 \mathcal{N}({\bf y})$ and $\mathcal{N}({\bf x})\subseteq_1 \mathcal{N}({\bf y})$.
\end{remark}

\begin{proposition}\label{NN}
 Let $\mathcal{N}$ be an $n$-dimensional fuzzy negation. Then, for all $i=1, \ldots, n$, the function $N_i: [0,1]\rightarrow [0,1]$ defined by 
\begin{eqnarray}\label{formN}
N_i(x)=\pi_i(\mathcal{N}(/{x}/))
\end{eqnarray}
is a fuzzy negation.
\end{proposition}

\begin{proof}
Trivially, $N_i(0)=\pi_i(\mathcal{N}(/0/)=\pi_i(/1/)=1$ and $N_i(1)=\pi_i(\mathcal{N}(/1/)=\pi_i(/0/)=0$. Let $x,y\in [0,1]$, then the following is verified:
\begin{eqnarray*}
x\leq y & \Rightarrow & /x/\leq /y/ 
           \Rightarrow \mathcal{N}(/x/)\geq \mathcal{N}(/y/) \ \ \mbox{ by \textbf{N2}}\\
           & \Rightarrow & \pi_i(\mathcal{N}(/x/))\geq \pi_i(\mathcal{N}(/y/)) \Rightarrow  N_i(x)\geq N_i(y) \ \mbox{by Eq. (\ref{mono_proji})}
\end{eqnarray*}
Therefore, Proposition~\ref{NN} holds.
\end{proof}\\

In the following, the necessary and sufficient conditions under which we can obtain $n$-representable fuzzy negation on $L_n([0,1])$ are discussed. 

\begin{theorem}\label{subseq monotone}
An $n$-dimensional fuzzy negation $\mathcal{N}$ is $n$-representable iff $\mathcal{N}$ is $\subseteq$-monotone.
\end{theorem}

\begin{proof}
$(\Rightarrow)$ If $\mathcal{N}$ is $n$-representable, then there exist fuzzy negations $N_1\leq\ldots\leq N_n$ such that $\mathcal{N}=\widetilde{N_1 \ldots N_n}$. 
For each $i=1,\ldots,n-1$, by the antitonicity of $N_{i's}$, it holds that
\begin{eqnarray*}
  \textbf{x}\subseteq_{n-i}\textbf{y} 
& \Rightarrow & \pi_{n-i}(\textbf{y})\leq \pi_{n-i}(\textbf{x})\leq \pi_{n-i+1}(\textbf{x})\leq  \pi_{n-i+1}(\textbf{y}) \ \mbox{ by Eq.(\ref{eq_pi})} \\
& \Rightarrow & N_{i}(\pi_{n-i+1}(\textbf{y}))\leq N_{i}(\pi_{n-i+1}(\textbf{x}))  \leq N_{i+1}  (\pi_{n-i}(\textbf{x})) \leq N_{i+1}(\pi_{n-i}(\textbf{y})) \\ &&  
 \mbox{ by $N2$ and since $N_{i}\leq N_{i+1}$}\\ 
& \Rightarrow & \pi_{i}(\mathcal{N}(\textbf{y}))\leq \pi_{i}(\mathcal{N}(\textbf{x}))\leq \pi_{i+1}(\mathcal{N}(\textbf{x}))\leq \pi_{i+1}(\mathcal{N}(\textbf{y})) ~ \mbox{by Eq.(\ref{tildeN})} \\
& \Rightarrow & \mathcal{N}(\textbf{x})\subseteq_i \mathcal{N}(\textbf{y}) \mbox{ by Eq.(\ref{eq_pi})}
\end{eqnarray*}
Hence, $\mathcal{N}$ is $\subseteq$-monotone. 

$(\Leftarrow)$ Firstly, for all $\textbf{x}\in L_n([0,1])$, when $i= 1, \dots ,n-1$, we have that $/x_{n-i+1}/\subseteq_{n-i+1} \textbf{x}$ as well as 
$/x_{n-i+1}/\subseteq_{n-i} \textbf{x}$. 
Since $\mathcal{N}$ is $\subseteq$-monotone then for each $i= 2, \dots ,n-1$
\begin{eqnarray*}
 /x_{n-i+1}/\subseteq_{n-i} \textbf{x}  & \Rightarrow & \mathcal{N}(/x_{n-i+1}/)\subseteq_{i} \mathcal{N}(\textbf{x}) \mbox{ by Eq.(\ref{cont_mono_i})}\\
                               & \Rightarrow &   \pi_i(\mathcal{N}(\textbf{x})) \leq \pi_i(\mathcal{N}(/x_{n-i+1}/))\ \mbox{by Eq.(\ref{eq_pi})} 
\end{eqnarray*}
and 
\begin{eqnarray*}
 /x_{n-i+1}/\subseteq_{n-i+1} \textbf{x}  & \Rightarrow & \mathcal{N}(/x_{n-i+1}/)\subseteq_{i-1} \mathcal{N}(\textbf{x}) \mbox{ by Eq.(\ref{cont_mono_i})}\\
                                & \Rightarrow &  \pi_i(\mathcal{N}(/x_{n-i+1}/)) \leq \pi_i(\mathcal{N}(\textbf{x})) \ \mbox{by Eq.(\ref{eq_pi})} 
 \end{eqnarray*}

So, for each $i= 2, \dots ,n-1$, $N_i(x_{n-i+1})=\pi_i(\mathcal{N}(/x_{n-i+1}/))=\pi_i(\mathcal{N}(\textbf{x}))$.

On the other hand, since $\mathcal{N}$ is $\subseteq$-monotone and decreasing, then 
\begin{eqnarray*}
/x_{n}/\subseteq_{n-1} \textbf{x} & \Rightarrow & \mathcal{N}(/x_{n}/)\subseteq_{1} \mathcal{N}(\textbf{x}) \mbox{ by Eq.(\ref{cont_mono_i})}\\
                               & \Rightarrow &   \pi_1(\mathcal{N}(\textbf{x})) \leq \pi_1(\mathcal{N}(/x_{n}/)) \mbox{ by Eq.(\ref{eq_pi})} 
\end{eqnarray*}
and 
 \begin{eqnarray*}
  \textbf{x} \leq /x_{n}/ & \Rightarrow & \mathcal{N}(/x_{n}/)\leq \mathcal{N}(\textbf{x})  \mbox{ by \textbf{N2}}\\
                                & \Rightarrow &   \pi_1(\mathcal{N}(/x_{n}/)) \leq\pi_1(\mathcal{N}(\textbf{x}))  \mbox{ by Eq.(\ref{mono_proji})} 
 \end{eqnarray*}
 
 Therefore,  $N_1(x_{n})=\pi_1(\mathcal{N}(/x_{n}/))=\pi_1(\mathcal{N}(\textbf{x}))$.  
Analogously,  since $\mathcal{N}$ is $\subseteq$-monotone and decreasing, then 
\begin{eqnarray*}
/x_{1}/\subseteq_{1} \textbf{x} & \Rightarrow & \mathcal{N}(/x_{1}/)\subseteq_{n-1} \mathcal{N}(\textbf{x}) \mbox{ by Eq.(\ref{cont_mono_i})}\\
                               & \Rightarrow &  \pi_n(\mathcal{N}(/x_{1}/))  \leq \pi_n(\mathcal{N}(\textbf{x})) \mbox{ by Eq.(\ref{eq_pi})} 
\end{eqnarray*}
and 
 \begin{eqnarray*}
   /x_{1}/\leq \textbf{x}  & \Rightarrow & \mathcal{N}(\textbf{x})\leq \mathcal{N}(/x_{1}/)  \mbox{ by \textbf{N2}}\\
                                & \Rightarrow &   \pi_n(\mathcal{N}(\textbf{x}))\leq \pi_n(\mathcal{N}(/x_{1}/))  \mbox{ by Eq.(\ref{mono_proji})} 
 \end{eqnarray*}
 
Concluding,  $N_n(x_{1})=\pi_n(\mathcal{N}(/x_{1}/))=\pi_n(\mathcal{N}(\textbf{x}))$.  
 So, $N_i(x_{n-i+1})=\pi_i(\mathcal{N}(\textbf{x}))$ for each $i=1,\ldots,n$ and consequently, $\mathcal{N}=\widetilde{N_1 \ldots N_n}$ and by Proposition \ref{NN}, the $N_{i's}$ are fuzzy negations and then $\mathcal{N}$ is n-representable.
\end{proof}

\begin{proposition}
If an $n$-dimensional fuzzy negation $\mathcal{N}$ is $n$-representable, then $\mathcal{N}$ is a function monotone by part.
\end{proposition}

\begin{proof}
 If $\mathcal{N}$ is $n$-representable, then there exist fuzzy negations $N_1\leq\ldots\leq N_n$ such that $\mathcal{N}=\widetilde{N_1 \ldots N_n}$. 
 Based on  the antitonicity of $N_{i's}$ and by property \textbf{N2},  if  $\pi_{n-i+1}(\textbf{x})\geq\pi_{n-i+1}(\textbf{y})$ for some $i=1, \ldots, n$ and 
 $\textbf{x},\textbf{y}\in L_n([0,1])$, then  we have that 
 $N_i(\pi_{n-i+1}(\textbf{x}))\leq N_i(\pi_{n-i+1}(\textbf{y}))$. Therefore, by Eq.(\ref{tildeN}), $\pi_i(\mathcal{N}(\textbf{x}))\leq \pi_i(\mathcal{N}(\textbf{y}))$. 
 Hence, 
 $\mathcal{N}$ is a monotone by part fuzzy negation on  $L_n([0,1])$.
\end{proof}

The partial order on fuzzy negations can be extended for $n$-dimensional fuzzy negations. For that, let $\mathcal{N}_1$ and $\mathcal{N}_2$ be $n$-dimensional fuzzy negations, then the following holds:
\begin{equation}\label{eq_tilde_neg}
\mathcal{N}_1\preceq \mathcal{N}_2 \mbox{ iff  for each } \textbf{x}\in L_n([0,1]), \mathcal{N}_1(\textbf{x})\leq \mathcal{N}_2(\textbf{x}).
\end{equation}

\begin{lemma}\label{preceq}
Let $N_1, \ldots, N_n$ be fuzzy negations. If $N_1\leq \ldots \leq N_n$, then $\widetilde{N_1}\preceq\widetilde{N_1 \ldots N_n}\preceq\widetilde{N_n}$.
\end{lemma}

\begin{proof}
By Eq.(\ref{tildeN}), we have that $\widetilde{N_1}(\textbf{x}) = (N_1(\pi_n({\bf x})), \ldots, N_1(\pi_1({\bf x})))$ and 
$\widetilde{N_1\ldots N_n}({\bf x})=(N_1(\pi_n({\bf x})), \ldots, N_n(\pi_1({\bf x})))$. Since $N_1\leq \ldots \leq N_n$, then 
$N_1(\pi_j(\textbf{x}))\leq N_{i}(\pi_j(\textbf{x}))$ with $i,j=1, \ldots, n$. So, by Eq.(\ref{eq_tilde_neg}), $\widetilde{N_1}\preceq\widetilde{N_1 \ldots N_n}$. Analogously we proof that $\widetilde{N_1 \ldots N_n}\preceq\widetilde{N_n}$.
\end{proof}

\begin{proposition}
Let $\mathcal{N}$ be an $n$-dimensional fuzzy negation. If $\mathcal{N}$ is $\subseteq$-monotone, then $\widetilde{N_\bot}\preceq \mathcal{N}\preceq \widetilde{N_\top}$.
\end{proposition}

\begin{proof}
Straightforward by Theorem \ref{subseq monotone} and Lemma \ref{preceq}.
\end{proof}

\begin{proposition}
Let $\mathcal{N}$ be an $n$-dimensional fuzzy negation. Then, it holds that $\mathcal{N}_\bot\preceq \mathcal{N}\preceq \mathcal{N}_\top$ whereas
\begin{eqnarray*}
\mathcal{N}_\bot({\bf x})=\left\{
\begin{array}{ccl}
/1/, & \mbox{if} & {\bf x}= /0/; \\
/0/, & \mbox{if} & {\bf x}\neq /0/;
\end{array}
\right. 
\end{eqnarray*}
\mbox { \hspace{0.5cm} and \hspace{0.5cm}}

\begin{eqnarray*}
\mathcal{N}_\top({\bf x})=\left\{
\begin{array}{ccl}
/0/ & \mbox{if}, & {\bf x}= /1/; \\
/1/ & \mbox{if}, & {\bf x}\neq /1/.
\end{array}
\right.
\end{eqnarray*}
\end{proposition}

\begin{proof}
If $\textbf{x}\neq /0/$, then $\mathcal{N}_\bot(\textbf{x})=/0/$ and so $\mathcal{N}_\bot(\textbf{x})\leq\mathcal{N}(\textbf{x})$. If $\textbf{x}=/0/$, then $\mathcal{N}(/0/)=/1/$ and so $\mathcal{N}_\bot(\textbf{x})\leq\mathcal{N}(\textbf{x})$. Analogously we proof when $\textbf{x}\neq /1/$ and $\textbf{x}= /1/$.
Therefore, $\mathcal{N}_\bot\preceq \mathcal{N}\preceq \mathcal{N}_\top$.
\end{proof}

\begin{remark}
Note that $\mathcal{N}_\bot\neq \widetilde{N}_\bot$ and $\mathcal{N}_\top\neq \widetilde{N}_\top$.
\end{remark}

\subsection{$n$-Dimensional strong fuzzy negations}

If an $n$-dimensional fuzzy negation $\mathcal{N}$ satisfies
  
\textbf{N3} $\mathcal{N}(\mathcal{N}(\textbf{x}))=\textbf{x}, \,\,\, \forall~ \textbf{x}\in L_n([0,1])$,\\
 it is called \textbf{$n$-dimensional strong fuzzy negation}. Additionally, an $n$-dimensional fuzzy negation $\mathcal{N}$ is \textbf{strict} if it is continuous and strictly decreasing, i.e., $\mathcal{N}(\textbf{x})<\mathcal{N}(\textbf{y})$ when $\textbf{y}<\textbf{x}$. 

\begin{proposition} Let $C_k(x) = \sqrt[k]{1-x^{k}}$, for $k \in \{1,2, \ldots , n\}$. Then 
\begin{equation}
\mathcal{C}({\bf{x}}) = \widetilde{C_k}({\bf x})=(C_k(\pi_n({\bf x})), \ldots, C_k(\pi_1({\bf x})))
\end{equation} 
is an $n$-representable strong fuzzy negation on $(L_n([0,1]))$.
\end{proposition}

\begin{proof}
Let $x\in [0,1]$, then 

\textbf{N1:}  $\widetilde{C_k}(/0/)=(C_k(\pi_n(/0/)), \ldots, C_k(\pi_1(/0/)))= (C_k(0), \ldots, C_k(0))=  /1/$\\
$\widetilde{C_k}(/1/)=(C_k(\pi_n(/1/)), \ldots, C_k(\pi_1(/1/)))= (C_k(1), \ldots, C_k(1))=  /0/$.

\textbf{N2:} Based on monotonicity of projection-functions,  if $\bf{x}\geq \bf{y}$ then $$(C_k(\pi_n({\bf x})), \ldots, C_k(\pi_1({\bf x}))) \leq (C_k(\pi_n({\bf y})), \ldots, C_k(\pi_1({\bf y}))),$$ 
\noindent therefore, $\widetilde{C_k}({\bf x}) \leq \widetilde{C_k}({\bf y})$ meaning that $\mathcal{C}(\bf{x}) \leq \mathcal{C}(\bf{y})$.

\textbf{N3:} For all $\textbf{x}\in L_n([0,1])$, it holds that
\begin{eqnarray*}
\mathcal{C}(\mathcal{C}(\textbf{x})) & = & \mathcal{C}(\widetilde{C_k}({\bf x}))= \mathcal{C}(C_k(\pi_n({\bf x})), \ldots, C_k(\pi_1({\bf x})))\\
                                                         & = & (C_k(C_k(\pi_1({\bf x}))), \ldots, C_k( C_k(\pi_n({\bf x})))) \\
                                                         & = & {\bf x}.
\end{eqnarray*}
Therefore, $\mathcal{C}\hspace{-0.1cm} = \hspace{-0.1cm}\widetilde{C_k}$ is an $n$-representable strong fuzzy negation on $L_n([0,1])$. 
\end{proof}

\begin{lemma}
\label{bijective} Let $\mathcal{N}$ be an $n$-dimensional fuzzy negation. If $\mathcal{N}$ is strong then it is bijective.
\end{lemma}

\begin{proof}
Trivially if $\mathcal{N}$ is strong then it is injective and surjective.
\end{proof}

\begin{proposition}\label{strongstrict}
Let $\mathcal{N}$ be an $n$-dimensional fuzzy negation. If $\mathcal{N}$ is strong then it is strict.
\end{proposition}

\begin{proof}
By Lemma \ref{bijective},  $\mathcal{N}$ is a strictly decreasing function. Therefore, if $\mathcal{N}$ is not continuous, then,  by the continuity of $\mathbb{R}^n$,
 there exists $\textbf{y} \in L_n([0,1])$ such that for all $\textbf{x} \in L_n([0,1])$ which is in contradiction with the Lemma \ref{bijective}.
\end{proof}

\begin{lemma}\label{veewedge}
Let ${\bf x}, {\bf y}\in L_n([0,1])$. Then, ${\bf x}\vee {\bf y}\in \mathcal{D}_n$ (${\bf x}\wedge {\bf y}\in \mathcal{D}_n$) iff either ${\bf x}\vee {\bf y}={\bf x}$ or ${\bf x}\vee {\bf y}={\bf y}$ (${\bf x}\wedge {\bf y}={\bf x}$ or ${\bf x}\wedge {\bf y}={\bf y}$). 
\end{lemma}

\begin{proof}
Straightforward from Eqs.(\ref{vee}) and (\ref{wedge}).
\end{proof}

\begin{proposition}\label{nveewedge}
 Let $\mathcal{N}$ be an $n$-dimensional strong fuzzy negation and ${\bf x}, {\bf y}\in L_n([0,1])$. Then, the following holds:
 \begin{enumerate}
 \item [(i)] $\mathcal{N}({\bf x})=/1/$ iff ${\bf x}=/0/$;
 \item [(ii)] $\mathcal{N}({\bf x})=/0/$ iff ${\bf x}=/1/$; 
 \item [(iii)] $\mathcal{N}({\bf x}\vee {\bf y})=\mathcal{N}({\bf x})\wedge\mathcal{N}({\bf y})$; 
  \item [(iv)] $\mathcal{N}({\bf x}\wedge {\bf y})=\mathcal{N}({\bf x})\vee\mathcal{N}({\bf y})$.
 \end{enumerate}
\end{proposition}

\begin{proof}
Items $(i)$ and $(ii)$  are straightforward by Lemma~\ref{bijective} and \textbf{N1}. $(iii)$~Let $\mathcal{N}$ be an $n$-dimensional strong fuzzy negation and 
$\textbf{x}, \textbf{y}\in L_n([0,1])$. Then, by antitonicity, $\mathcal{N}(\textbf{x}\vee\textbf{y})\leq \mathcal{N}(\textbf{x})$ and 
$\mathcal{N}(\textbf{x}\vee\textbf{y})\leq \mathcal{N}(\textbf{y})$. So, $\mathcal{N}(\textbf{x}\vee\textbf{y})\leq \mathcal{N}(\textbf{x})\wedge\mathcal{N}(\textbf{y})$. 
Suppose that $\mathcal{N}(\textbf{x}\vee\textbf{y})< \mathcal{N}(\textbf{x})\wedge\mathcal{N}(\textbf{y})$.
Then there exists $\textbf{z}\in L_n([0,1])$ such that $\mathcal{N}(\textbf{x}\vee\textbf{y})< \textbf{z}< \mathcal{N}(\textbf{x})\wedge \mathcal{N}(\textbf{y})$ and 
therefore, $\textbf{z}< \mathcal{N}(\textbf{x})$ and $\textbf{z}< \mathcal{N}(\textbf{y})$. So, by Proposition \ref{strongstrict}, 
$\mathcal{N}(\textbf{z})> \mathcal{N}(\mathcal{N}(\textbf{x}))$ and $\mathcal{N}(\textbf{z})>\mathcal{N}(\mathcal{N}(\textbf{y}))$. Since $\mathcal{N}$ is strong we have 
that $\mathcal{N}(\textbf{z})> \textbf{x}$ and $\mathcal{N}(\textbf{z})> \textbf{y}$. Hence, $\mathcal{N}(\textbf{z})\geq \textbf{x}\vee \textbf{y}$. Thus, by \textbf{N2} and 
\textbf{N3}, $\textbf{z}\leq \mathcal{N}(\textbf{x}\vee\textbf{y})$ which is a contradiction and therefore 
$\mathcal{N}({\bf x}\vee {\bf y})=\mathcal{N}({\bf x})\wedge\mathcal{N}({\bf y})$. $(iv)$ Analogous to the above prove of item $(iii)$.
\end{proof}

%
%
%
%
%

\begin{lemma}\label{/x/}
Let $\mathcal{N}$ be an $n$-dimensional strong fuzzy negation. If for a  ${\bf x}\not\in D_n$ we have that  $\mathcal{N}({\bf x})\in D_n$ then for some $j=1,\ldots,n-1$:
\begin{equation}\label{eq-00011111} 
\mathbf{x}=(0^{(j)},1^{(n-j)}),
\end{equation}
where $(0^{(j)},1^{(n-j)})$ denotes $(\underbrace{0,\ldots,0}_{j-times},\underbrace{1,\ldots,1}_{(n-j)-times})$.
\end{lemma}

\begin{proof}
Suppose that for some $\textbf{x}\not\in D_n$, there exists $z\in [0,1]$ such that $\mathcal{N}(\textbf{x})=/z/$. By Proposition~\ref{nveewedge}, $z\in (0,1)$ and 
$\displaystyle\bigwedge_{i=1}^{n} \mathcal{N}(\widetilde{x_i})=/z/$ 
once $\textbf{x}=\displaystyle\bigvee_{i=1}^{n} \widetilde{x_i}$, where $\widetilde{x_i}=(0^{(i-1)}, x_i^{(n-i+1)})$. 
So by Lemma \ref{veewedge}, $/z/=\mathcal{N}(\widetilde{x_k})$ for some 
$k=1, \ldots, n$. But, once $\mathcal{N}$ is bijective, then $\textbf{x}=\widetilde{x_k}$.

Analogously, since 
$\textbf{x}=\displaystyle\bigwedge_{i=1}^{n} \widehat{x_i}$, where $\widehat{x_i}=(x_i^{(i)},1^{(n-i)})$, then by Proposition~\ref{nveewedge} we have that 
$\displaystyle\bigvee_{i=1}^{n} \mathcal{N}(\widehat{x_i})=/z/$. So by Lemma \ref{veewedge}, $/z/=\mathcal{N}(\widetilde{x_j})$ for some 
$j=1, \ldots, n$. But, once $\mathcal{N}$ is bijective, then $\textbf{x}=\widehat{x_j}$. Hence, $\widetilde{x_k}=\widehat{x_j}$ and consequently $k=j+1$. Therefore, the Equation (\ref{eq-00011111}) holds. 
\end{proof}

%

\begin{theorem}\label{th-/x/}
Let $\mathcal{N}$ be an $n$-dimensional strong fuzzy negation. Then for each ${\bf x}\not\in D_n$, $\mathcal{N}({\bf x})\not\in D_n$.
\end{theorem}

\begin{proof}
Let $J=\{j\in\{1,\ldots,n-1\} : \mathcal{N}(\mathbf{x}_j)\in \mathcal{D}_n\}$ where $\mathbf{x}_j=(0^{(j)},1^{(n-j)})$. 
Observe that if $j\leq i$ then $\mathbf{x}_i\leq \mathbf{x}_j$. 
If  $\mathcal{N}({\bf x})\not\in D_n$ for some  $\textbf{x}\not\in D_n$, then by Lemma \ref{/x/}, $J$ is a finite and not empty set. 
Let $j= \min J$ and $z\in (0,1)$ such that $ \mathcal{N}(\mathbf{x}_j)=/z/$. 
For each $y\in (0,z)$ we have that $/y/</z/$ and so, because $\mathcal{N}$ is strong, $\mathbf{x}_j=\mathcal{N}(/z/)<\mathcal{N}(/y/)$. Therefore, 
since $(0^{(j)},1^{(n-j)})<\mathcal{N}(/y/)$, we have that $\mathcal{N}(/y/)=(a_1,\ldots,a_j,1^{(n-j)})$  for some $a_1,\ldots,a_j\in [0,1]$. On the other hand, since for each $i\in J$ we have that $\mathbf{x}_i\leq \mathbf{x}_j$, then by Lemma \ref{/x/},  $\mathcal{N}(/y/)\in \mathcal{D}_n$. 
Therefore, $\mathcal{N}(/y/)= /1/$ which is a contradiction with Proposition \ref{nveewedge}.
\end{proof}

\begin{corollary}\label{coro-3.3}
If $\mathcal{N}$ is an $n$-dimensional strong fuzzy negation then $\mathcal{N}$ satisfies \textbf{DP}.
\end{corollary}

\begin{proof}
Suppose that there exists $/x/\in \mathcal{D}_n$ such that $\textbf{y}=\mathcal{N}(/x/)\not\in \mathcal{D}_n$. Since $\mathcal{N}$ is strong,  $\mathcal{N}(\textbf{y})=/x/$, i.e., $\mathcal{N}$ map a non-degenerate  in a degenerate element which is a contradiction by Theorem 
\ref{th-/x/}.
\end{proof}

\begin{lemma}\cite[Theorem 3.2]{bedregal12}\label{bedregal12}
Let $\mathcal{N}:L_n([0,1])\rightarrow L_n([0,1])$. $\mathcal{N}$ is an $n$-dimensional strong fuzzy negation satisfying the property $\textbf{DP}$ iff there exists a strong fuzzy negation $N$ such that $\mathcal{N}=\widetilde{N}$.
\end{lemma}

\begin{theorem}\label{N=N}
 $\mathcal{N}$ is an $n$-dimensional strong fuzzy negation iff there exists a strong fuzzy negation $N$ such that $\mathcal{N}=\widetilde{N}$.
\end{theorem}

\begin{proof}
Straightforward from Corollary \ref{coro-3.3} and Lemma \ref{bedregal12}.
\end{proof}

\subsection{$n$-Dimensional equilibrium points}

Analogous to fuzzy negations, we will define an $n$-dimensional equilibrium point as the following:

\begin{definition}
An element ${\bf e}\in L_n([0,1])$ is an $n$-dimensional equilibrium point for an $n$-dimensional fuzzy negation $\mathcal{N}$ if $\mathcal{N}({\bf e})= {\bf e}$.
\end{definition}

\begin{remark}
Let $\mathcal{N}$ be a strict $n$-dimensional fuzzy negation. If  ${\bf x}<{\bf e}$ then $\mathcal{N}({\bf x})>{\bf e}$ and if ${\bf e}<{\bf x}$ then $\mathcal{N}({\bf x})<{\bf e}$.
\end{remark}

\begin{proposition}\label{eqpoint}
Let $N$ be a fuzzy negation with the equilibrium point $e$. Then, $/e/$ is an $n$-dimensional equilibrium point of $\widetilde{N}$.
\end{proposition}

\begin{proof}
Straightforward.
\end{proof}

\begin{corollary}
Let $\mathcal{N}$ be an $n$-dimensional strong fuzzy negation. Then, there exists an element $/e/\in \mathcal{D}_n$ such that $/e/$ is an $n$-dimensional equilibrium point of $\mathcal{N}$. 
\end{corollary}

\begin{proof}
Straightforward from Corollary \ref{coro-3.3}, Theorem \ref{N=N}, and Proposition \ref{eqpoint}.
\end{proof}

\begin{corollary}\label{eqpoint1}
Let $N$ be the strong fuzzy negation and $e\in (0,1)$. Then, $/e/$ is an $n$-dimensional equilibrium point of $\widetilde{N}$ iff $e$ is an equilibrium point of $N$.
\end{corollary}

\begin{proof}
Straightforward from Theorem \ref{N=N}, and Proposition \ref{eqpoint}.
\end{proof}

\begin{remark} For $k \in \{1,2, \ldots , n\}$, consider the strong fuzzy negation given by Eq.(\ref{eq_ckn}) in Example~\ref{ex:1}, $C_k(x) =\sqrt[n-k+1]{1-x^{n-k+1}}$ and its corresponding equilibrium point $\sqrt[n-k+1]{0.5}$. 
By taking $\mathbf{x} =(\sqrt[n-k+1]{0.5}, \sqrt[n-k+1]{0.5}, \ldots , \sqrt[n-k+1]{0.5}) \in L_n([0,1])$, one can observe that
\begin{eqnarray*}
\mathcal{C}(\mathbf{x}) & = &  \widetilde{C_k}({\bf x})=(C_k(\pi_n({\bf x})), \ldots , C_k(\pi_1({\bf x}))) \\ &  =&  (\sqrt[n-k+1]{0.5},\sqrt[n-k+1]{0.5},  \ldots ,\sqrt[n-k+1]{0.5})\\ & =& \mathbf{x},
\end{eqnarray*}
meaning that such operator preserves distinct equilibrium point of component-functions $C_k$ of an $n$-representable fuzzy negation $\mathcal{C} = \widetilde{C_k}$.
\end{remark}

\begin{remark} For $k \in \{1,2, \ldots , n\}$, consider the fuzzy negation given by Eq.(\ref{eq-dk}) in Proposition~\ref{pro-dk}, $C^k(x) = 1-x^k$ and its corresponding equilibrium point $e^k$. By N2, $e^1 \leq\ldots\leq e^n$, then $\mathbf{x} =(e^1, \ldots , e^n) \in L_n([0,1])$.  So, it is immediate  observing that
\begin{eqnarray*}
\mathcal{C}(\mathbf{x}) &=&  \widetilde{C^1 \ldots C^n}({\bf x})=(C^1(\pi_n({\bf x})), \ldots , C^n(\pi_1({\bf x}))) \\ &  = & (e^1,  \ldots ,  e^n)\\& =& \mathbf{x},
\end{eqnarray*}
meaning that such operator preserves the equilibrium points of component-functions $C^1,\ldots, C^2$ of the $n$-representable fuzzy negation $\mathcal{C} = \widetilde{C^1\ldots C^n}$.
\end{remark}

\section{$n$-Dimensional automorphisms}\label{sec-4}

In this section we briefly recall some well-known results of automorphism on $L([0,1])$, in order to extend them to the $n$-dimensional approach, mainly connected to representable fuzzy negation on $L_n([0,1])$. The notion of $\mathcal{N}$-preserving $n$-dimensional fuzzy automorphism is also discussed.

In \cite{bedregal12}, an $n$-dimensional automorphism is defined as follows:

\begin{definition}
A function $\varphi: L_n([0,1]) \rightarrow L_n([0,1])$ is an $n$-dimensional automorphism if $\varphi$ is bijective and the following condition is satisfied
\begin{center}
${\bf x}\leq {\bf y}$ iff $\varphi({\bf x})\leq \varphi({\bf y})$.
\end{center}
\end{definition}

\begin{theorem}\cite[Theorem 3.4]{bedregal12}\label{bedregal3.4}
Let $\varphi: L_n([0,1]) \rightarrow L_n([0,1])$. A function $\varphi \in Aut(L_n([0,1]))$ iff there exists $\psi \in Aut([0,1])$ such that
\begin{center}
$\varphi(\pi_1({\bf x}), \ldots, \pi_n({\bf x}))=(\psi(\pi_1({\bf x})), \ldots, \psi(\pi_n({\bf x})))$.
\end{center}
In this case, we will denote $\varphi$ by $\widetilde{\psi}$. Thus, the following holds
\begin{eqnarray}\label{widetildepsi}
\widetilde{\psi}(\pi_1({\bf x}), \ldots, \pi_n({\bf x}))=(\psi(\pi_1({\bf x})), \ldots, \psi(\pi_n({\bf x}))).
\end{eqnarray}
\end{theorem}

\begin{corollary}\label{corvarphi}
If $\varphi \in Aut(L_n([0,1]))$ then it is continuous, strictly increasing, $\varphi(/0/)=/0/$ and $\varphi(/1/)=/1/$.
\end{corollary} 

\begin{proposition}\cite[Proposition 3.4]{bedregal12}\label{tildepsi}
Let $\psi \in Aut([0,1])$. Then, the following holds:
\begin{eqnarray}\label{psi_reverse}
\widetilde{\psi^{-1}}&=\widetilde{\psi}^{-1}.
\end{eqnarray}
\end{proposition}

\begin{proposition}
Let $\varphi \in Aut(L_n([0,1]))$. $\mathcal{N}$ is $n$-dimensional (strict, strong) fuzzy negation  iff $\mathcal{N}^{\varphi}$ is an $n$-dimensional (strict, strong) fuzzy negation such that,  for all ${\bf x}\in L_n([0,1])$, the following holds:
\begin{center}
$\mathcal{N}^{\varphi}({\bf x})=\varphi^{-1}(\mathcal{N}(\varphi({\bf x})))$.
\end{center}
\end{proposition}

\begin{proof}
$(\Rightarrow)$ \textbf{N1}: Let $\mathcal{N}$ be an $n$-dimensional fuzzy negation. Then,  the following holds:
\begin{eqnarray*}
\mathcal{N}^{\varphi}(/0/)& = & \varphi^{-1}(\mathcal{N}(\varphi(/0/))) = \varphi^{-1}(\mathcal{N}(/0/)) \\ & = & \varphi^{-1}(/1/) = /1/. 
\end{eqnarray*}
Analogously we proof that $\mathcal{N}^{\varphi}(/1/)=/0/$.

\textbf{N2}: If $\textbf{x}\leq \textbf{y}$ then $ \varphi(\textbf{x})\leq \varphi(\textbf{y})$ and the following holds: 
\begin{eqnarray*}
\mathcal{N}(\varphi(\textbf{x}))\geq \mathcal{N}(\varphi(\textbf{y})) & \Rightarrow & \varphi^{-1}(\mathcal{N}(\varphi(\textbf{x})))\geq \varphi^{-1}(\mathcal{N}(\varphi(\textbf{y}))) \\ 
& \Rightarrow & \mathcal{N}^{\varphi}({\bf x})\geq \mathcal{N}^{\varphi}({\bf y}).
\end{eqnarray*}
Therefore, $\mathcal{N}^{\varphi}$ is an $n$-dimensional fuzzy negation. In addition, if $\mathcal{N}$ is strictly decreasing then, as the $n$-dimensional automorphism, trivially, $\mathcal{N}^{\varphi}$ is strictly decreasing. If $\mathcal{N}$ is continuous then, by Corollary \ref{corvarphi}, $\varphi$ and $\varphi^{-1}$ are continuous. Since the composition of the continuous function is continuous, then $\mathcal{N}^{\varphi}$ is continuous. Thus, if $\mathcal{N}$ is strict, then $\mathcal{N}^{\varphi}$ is also strict. Moreover, if $\mathcal{N}$ is an $n$-dimensional strong fuzzy negation then the following holds:
\begin{eqnarray*}
\mathcal{N}^{\varphi}(\mathcal{N}^{\varphi}(\textbf{x})) & = & \mathcal{N}^{\varphi}(\varphi^{-1}(\mathcal{N}(\varphi(\textbf{x})))) \\
& = &\varphi^{-1}(\mathcal{N}(\varphi(\varphi^{-1}(\mathcal{N}(\varphi(\textbf{x})))))) \\
& = & \varphi^{-1}(\mathcal{N}(\mathcal{N}(\varphi(\textbf{x}))))\\ & = & \varphi^{-1}(\varphi(\textbf{x})) \\ & = & \textbf{x}.
\end{eqnarray*}

$(\Leftarrow)$ Let $\mathcal{N}^{\varphi}$ be an $n$-dimensional (strict, strong) fuzzy negation. By the above proof, $(\mathcal{N}^{\varphi})^{\varphi^{-1}}$ also is an $n$-dimensional (strict, strong) fuzzy negation. Since $ \mathcal{N}=(\mathcal{N}^{\varphi})^{\varphi^{-1}}$, then $\mathcal{N}$ is an $n$-dimensional (strict, strong) fuzzy negation. 
\end{proof}\\

\begin{proposition}\label{pro-auto-ntilde}
Let $N_1\ldots N_n$ be fuzzy negations and $\psi$ be an automorphism. Then, the following holds:
\begin{center}
$\widetilde{N_1\ldots N_n}^{\widetilde{\psi}}=\widetilde{N_1^{\psi}\ldots N_n^{\psi}}$.
\end{center}  
\end{proposition}

\begin{proof} For all ${\bf x}\in L_n([0,1])$, it holds that:
\begin{eqnarray*}
\widetilde{N_1^{\psi}\ldots N_n^{\psi}}(\textbf{x}) & = & (N_1^{\psi}(\pi_n(\textbf{x})), \ldots, N_n^{\psi}(\pi_1(\textbf{x}))) \  \mbox{by Eq.(\ref{rhoN})} \\
                                                                              & = & (\psi^{-1}(N_1(\psi(\pi_n(\textbf{x})))), \ldots, \psi^{-1}( N_n(\psi(\pi_1(\textbf{x}))))) \ \mbox{by Eq.(\ref{widetildepsi})} \\
                                                                              & = & \widetilde{\psi^{-1}}(N_1(\psi(\pi_n(\textbf{x}))), \ldots,  N_n(\psi(\pi_1(\textbf{x})))) \ \mbox{by Eq.(\ref{tildeN})}\\
                                                                              & = & \widetilde{\psi^{-1}}(\widetilde{N_1, \ldots, N_n}(\psi(\pi_1(\textbf{x})), \ldots, \psi(\pi_n(\textbf{x})))) \ \mbox{by Eq.(\ref{widetildepsi})}\\
                                                                             & = & \widetilde{\psi^{-1}}(\widetilde{N_1, \ldots, N_n}(\widetilde{\psi}(\pi_1(\textbf{x}), \ldots, \pi_n(\textbf{x}))))  \ \mbox{by Eq.(\ref{psi_reverse})}\\
                                                                             & = & \widetilde{\psi}^{-1}(\widetilde{N_1, \ldots, N_n}(\widetilde{\psi}(\pi_1(\textbf{x}), \ldots, \pi_n(\textbf{x}))))  \ \ \mbox{by Prop. \ref{tildepsi}} \\
                                                                             & = & \widetilde{N_1\ldots N_n}^{\widetilde{\psi}}(\textbf{x}).                                                                           
\end{eqnarray*}
 Therefore, Proposition~\ref{pro-auto-ntilde} is verified.
\end{proof}

\begin{lemma}\cite[Corollary 3.1]{bedregal12}\label{corbedregal12}
A function $\mathcal{N}: L_n([0,1])\rightarrow L_n([0,1])$ is an $n$-dimensional strong fuzzy negation satisfying the property $\textbf{DP}$ iff there exists an $n$-dimensional automorphism $\varphi$ such that $\mathcal{N}=\mathcal{N}^{\varphi}_{S}$, where 
\begin{center}
$\mathcal{N}_S({\bf x})=(1-\pi_n({\bf x}), 1-\pi_{n-1}({\bf x}), \ldots, 1-\pi_1({\bf x}))$.
\end{center}
\end{lemma}

The following theorem is a generalization of the Trillas theorem \cite{trillas79} and a generalization for the interval case given by Bedregal in \cite{bedregal10}.

\begin{theorem}\label{theondsfn-nda}
A function $\mathcal{N}: L_n([0,1])\rightarrow L_n([0,1])$ is an $n$-dimensional strong fuzzy negation iff there exists an $n$-dimensional automorphism $\varphi$ such that $\mathcal{N}=\mathcal{N}^{\varphi}_{S}$.
\end{theorem}

\begin{proof}
Straightforward from Corollary \ref{coro-3.3} and Lemma \ref{corbedregal12}. 
\end{proof}

\begin{proposition}\label{n<nvarphi}
Let $\mathcal{N}$ be an $n$-dimensional strict (strong) fuzzy negation and the $n$-dimensional automorphism $\varphi({\bf x})={\bf x}^2$, i.e., $\varphi({\bf x})=((\pi_1({\bf x}))^2, \ldots,$ $ (\pi_n({\bf x}))^2)$. Then, $\mathcal{N}<\mathcal{N}^{\varphi}$ and $(\mathcal{N}^{\varphi})^{-1}<\mathcal{N}$.
\end{proposition}

\begin{proof}
Clearly, $\varphi^{-1}(\textbf{x})=\sqrt{\textbf{x}}$, i.e., $\varphi^{-1}(\textbf{x})=(\sqrt{\pi_1(\textbf{x})}, \ldots, \sqrt{\pi_n(\textbf{x})})$. 
Since $\textbf{x}^2<\textbf{x}$ for each $\textbf{x}\in L_n([0,1]) - \{/0/, /1/\}$, then because $\mathcal{N}$ is strict we have that 
$\mathcal{N}(\textbf{x})<\mathcal{N}(\textbf{x}^2)$. So, 
$\varphi^{-1}(\mathcal{N}(\textbf{x}))<\varphi^{-1}(\mathcal{N}(\varphi(\textbf{x})))=\mathcal{N}^\varphi(\textbf{x})$. But, since $\textbf{x}<\sqrt{\textbf{x}}$ 
for each $\textbf{x}\in L_n([0,1])-\{/0/, /1/\}$, then $\mathcal{N}(\textbf{x})<\sqrt{\mathcal{N}(\textbf{x})}$. Therefore, 
$\mathcal{N}(\textbf{x})<\mathcal{N}^\varphi(\textbf{x})$. Analogously we proof that $(\mathcal{N}^{\varphi})^{-1}<\mathcal{N}$.
\end{proof}

\begin{corollary}
There exists neither a lesser nor greater $n$-dimensional strict (strong) fuzzy negation.
\end{corollary}

\begin{proof}
Straightforward from Proposition \ref{n<nvarphi}.
\end{proof}

\subsection{$\mathcal{N}$-Preserving $n$-dimensional automorphisms}

Let $\mathcal{N}$ be an $n$-dimensional fuzzy negation. An $n$-dimensional automorphism $\varphi$ is \textbf{$\mathcal{N}$-preserving $n$-dimensional automorphism} if, for each $\textbf{x}\in L_n([0,1])$, the following holds
\begin{eqnarray}\label{Npreservingfor}
\varphi(\mathcal{N}(\textbf{x}))=\mathcal{N}(\varphi(\textbf{x})).
\end{eqnarray}

The following theorem shows us that $\mathcal{N}$-Preserving $n$-dimensional automorphisms are strongly related with the notion of $N$-preserving automorphisms.\\

\begin{theorem}\label{Npreserving}
Let $\varphi$ be an $n$-dimensional automorphism, $\mathcal{N}$ be a representable $n$-dimensional fuzzy negation, $\psi$ be the automorphism such that $\varphi=\widetilde{\psi}$ and $N_1, \ldots, N_n$ be fuzzy negations such that $\mathcal{N}=\widetilde{N_1\ldots N_n}$. Then, $\varphi$ is a $\mathcal{N}$-preserving $n$-dimensional automorphism iff $\psi$ is an $N_i$-preserving  automorphism, for each $i=1, \ldots, n$.
\end{theorem}

\begin{proof}
($\Rightarrow$) Let $x\in [0,1]$, then
\begin{eqnarray*}
\psi(N_i(x)) & = & \psi(\pi_i(\mathcal{N}(/x/)))  \mbox{ by Eq.(\ref{formN})}\\
                  & = &\pi_i(\widetilde{\psi}(\mathcal{N}(/x/)))   \mbox{ by Eq.(\ref{widetildepsi})}\\
                  & = &\pi_i(\mathcal{N}(\widetilde{\psi}(/x/)))   \mbox{ by Eq.(\ref{Npreservingfor})}\\
                  & = &\pi_i(\mathcal{N}(/\psi(x)/))   \mbox{ by Eq.(\ref{widetildepsi})}\\
                  & = & N_i(\psi(x))  \mbox{ by Eq.(\ref{formN})}
\end{eqnarray*}

($\Leftarrow$) Let $\textbf{x}\in L_n([0,1])$, then
\begin{eqnarray*}
\widetilde{\psi}(\mathcal{N}(\textbf{x})) & = & \widetilde{\psi}(\widetilde{N_1\ldots N_n}(\textbf{x})) \\
                                                              & = & \widetilde{\psi}(N_1(\pi_n(\textbf{x})), \ldots,  N_n(\pi_1(\textbf{x}))) \mbox{ by Eq.(\ref{tildeN})}\\
                                                              & = & (\psi(N_1(\pi_n(\textbf{x}))), \ldots, \psi( N_n(\pi_1(\textbf{x})))) \mbox{ by Eq.(\ref{widetildepsi})}\\
                                                              & = & (N_1(\psi(\pi_n(\textbf{x}))), \ldots,  N_n(\psi(\pi_1(\textbf{x})))) \mbox{ by Eq.(\ref{Npreservingfor})}\\
                                                              & = & \widetilde{N_1 \ldots  N_n}(\psi(\pi_1(\textbf{x})),\ldots,\psi(\pi_n(\textbf{x}))) \mbox{ by Eq.(\ref{tildeN})}\\
                                                              & = & \widetilde{N_1 \ldots  N_n}(\widetilde{\psi}(\textbf{x}))\mbox{ by Eq.(\ref{widetildepsi})}\\
                                                              & = & \mathcal{N}(\widetilde{\psi}(\textbf{x}))
\end{eqnarray*}
Therefore, Proposition~\ref{Npreserving} holds.
\end{proof}

The following theorem is an $n$-dimensional version of Proposition \ref{bedregalprop2.6} which extends \cite[Proposition 4.2]{navara99} and \cite[Proposition 7.5]{bedregal10} for interval case. It provides an expression for all  $\mathcal{N}$-preserving $n$-dimensional automorphisms in $L_n([0,1])$.

\begin{theorem}\label{theoremvarphiN}
Let $\mathcal{N}$ be an $n$-dimensional strong fuzzy negation with $/e/$ as the degenerate equilibrium point and $\varphi$ be an $n$-dimensional automorphism on  $L_n([0,e])=\{{\bf x}\in L_n([0,1]): \pi_n({\bf x})\leq e\}$.\footnote{All definitions and results described in the beginning of this section (until Proposition 5.1) can be adapted for $L_n([0,e])$.} Then, $\varphi^{\mathcal{N}}: L_n([0,1])\rightarrow L_n([0,1])$ defined by

\begin{eqnarray}\label{varphiN}
\varphi^{\mathcal{N}}({\bf x}){=}\left\{
\begin{array}{ll}
\varphi({\bf x}) ~ \hspace{1.15cm}  \mbox{if} \ {\bf x}\leq /e/ \\
\mathcal{N}(\varphi(\mathcal{N}({\bf x})))  ~\mbox{if} \  {\bf x}> /e/ \\ 
(\pi_1(\varphi({\bf x})){,} \ldots{,} \pi_i(\varphi({\bf x})){,}   \pi_{i+1}(\mathcal{N}(\varphi(\mathcal{N}({\bf x})))){,}\ldots{,}  \pi_{n}(\mathcal{N}(\varphi(\mathcal{N}({\bf x}))))) ~  \mbox{if} \ \pi_i({\bf x})\leq e<\pi_{i+1}({\bf x})
\end{array}
\right.
\end{eqnarray}
is an $\mathcal{N}$-preserving $n$-dimensional automorphism. 
\end{theorem}

\begin{proof}
By Theorem~\ref{bedregal3.4}, there exists an automorphism $\psi$ such that  $\varphi=\widetilde{\psi}$. Analogously, by Theorem~\ref{N=N}, there exists a strong fuzzy negation $N$ such that  $\mathcal{N}=\widetilde{N}$. Thus, it holds that

If $\textbf{x}=/e/$ and since $/e/=\mathcal{N}(/e/)$, then $\mathcal{N}(/e/)=\mathcal{N}(\textbf{x})$. So,
\begin{eqnarray*}
\varphi^{\mathcal{N}}(\mathcal{N}(\textbf{x})) & = & \varphi(\mathcal{N}(\textbf{x})) \mbox{ since $\mathcal{N}(\textbf{x})= /e/$} \\
                                                                        & = & \mathcal{N}(\varphi(\textbf{x}) \mbox{ by Eq.~(\ref{Npreservingfor})} \\
                                                                        & = & \mathcal{N}(\varphi^{\mathcal{N}}({\bf x})) \mbox{ since $\textbf{x}=/e/$}
\end{eqnarray*}
If $\textbf{x}</e/$, then since $\mathcal{N}$ is strict, $/e/=\mathcal{N}(/e/)<\mathcal{N}(\textbf{x})$ and so,
\begin{eqnarray*}
\varphi^{\mathcal{N}}(\mathcal{N}(\textbf{x})) & = & \mathcal{N}(\varphi(\mathcal{N}(\mathcal{N}(\textbf{x})))) \mbox{ since $\mathcal{N}(\textbf{x})> /e/$} \\
                                                                        & = & \mathcal{N}(\varphi(\textbf{x})) \mbox{ since $\mathcal{N}$ is strong} \\
                                                                        & = & \mathcal{N}(\varphi^{\mathcal{N}}({\bf x})) \mbox{ since $\textbf{x}</e/$}
\end{eqnarray*}
If $\textbf{x}>/e/$ then since $\mathcal{N}$ is strict, $\mathcal{N}(\textbf{x})</e/$ implying new  results as follows
\begin{eqnarray*}
\varphi^{\mathcal{N}}(\mathcal{N}(\textbf{x})) & = & \varphi(\mathcal{N}(\textbf{x})) = \mathcal{N}(\mathcal{N}(\varphi(\mathcal{N}(\textbf{x})))) \  \mbox{since $\mathcal{N}$ is strong}.\\
 & = & \mathcal{N}(\varphi^{\mathcal{N}}({\bf x})) \ \mbox{ since $\textbf{x}>/e/$}
\end{eqnarray*}
If $\pi_i({\bf x})< e<\pi_{i+1}({\bf x})$ then $N(\pi_{i+1}({\bf x}))< N(e)< N(\pi_{i}({\bf x}))$ and   by Corollary \ref{eqpoint1}, $N(e)=e$.

In addition,
\begin{eqnarray}\label{NpsiN}
&   (\pi_1(\varphi({\bf x})){,} \ldots{,} \pi_i(\varphi({\bf x})){,} \pi_{i+1}(\mathcal{N}(\varphi(\mathcal{N}({\bf x})))){,}\ldots{,} \pi_{n}(\mathcal{N}(\varphi(\mathcal{N}({\bf x})))))  \\  
& =   (\psi(\pi_1({\bf x})), \ldots, \psi(\pi_i({\bf x})), N(\pi_{n-i}(\varphi(\mathcal{N}({\bf x})))),\ldots{,} N(\pi_{1}(\varphi(\mathcal{N}({\bf x})))))  \nonumber \\
& =   (\psi(\pi_1({\bf x})), \ldots, \psi(\pi_i({\bf x})),  N(\psi(\pi_{n-i}(\mathcal{N}({\bf x})))),\ldots{,} N(\psi(\pi_{1}(\mathcal{N}({\bf x})))))  \nonumber \\
& =   (\psi(\pi_1({\bf x})), \ldots, \psi(\pi_i({\bf x})),  N(\psi(N(\pi_{i+1}({\bf x})))){,}\ldots, N(\psi(N(\pi_{n}({\bf x}))))). \nonumber
\end{eqnarray}

So, the following holds 
\begin{eqnarray*}
& &\varphi^{\mathcal{N}}(\mathcal{N}(\textbf{x}))\\
& = & (\psi(\pi_1(\mathcal{N}(\textbf{x}))), \ldots, \psi(\pi_{n-i}(\mathcal{N}(\textbf{x}))),  N(\psi(N(\pi_{n-i+1}(\mathcal{N}(\textbf{x}))))), \ldots,  N(\psi(N(\pi_{n}(\mathcal{N}(\textbf{x})))))) \\&&  \mbox{ by Eq.(\ref{NpsiN})} \\
 & = & (\psi(N(\pi_n(\textbf{x}))), \ldots, \psi(N(\pi_{i+1}(\textbf{x}))),  N(\psi(N(N(\pi_{i}(\textbf{x}))))),\ldots, N(\psi(N(N(\pi_1(\textbf{x})))))) \\ 
& = & (N(N(\psi(N(\pi_n(\textbf{x}))))), \ldots, N(N(\psi(N(\pi_{i+1}(\textbf{x}))))), N(\psi(\pi_{i}(\textbf{x}))),\ldots, N(\psi(\pi_1(\textbf{x})))\\ && \mbox{ since $N$ is strong} \\
& = & \widetilde{N}(\psi(\pi_1(\textbf{x})), \ldots, \psi(\pi_{i}(\textbf{x})),N(\psi(N(\pi_{i+1}(\textbf{x})))), \ldots,  N(\psi(N(\pi_n(\textbf{x})))))  \mbox{ by Eq.(\ref{tildeN})} \\
& = & \mathcal{N}(\psi(\pi_1(\textbf{x})), \ldots, \psi(\pi_{i}(\textbf{x})),N(\psi(N(\pi_{i+1}(\textbf{x})))), \ldots, N(\psi(N(\pi_n(\textbf{x}))))) \\ 
& = & \mathcal{N}(\pi_1(\varphi({\bf x})), \ldots, \pi_{i}(\varphi({\bf x})), \pi_{i+1}(\mathcal{N}(\varphi(\mathcal{N}({\bf x})))),\ldots, \pi_{n}(\mathcal{N}(\varphi(\mathcal{N}({\bf x}))))) ~   \mbox{by Eq.(\ref{NpsiN}),}  \\ && \mbox{based on results of Theorems \ref{bedregal3.4} and~\ref{N=N}} \\
& = & \mathcal{N}(\varphi^{\mathcal{N}}(\textbf{x})).
\end{eqnarray*}

If $\pi_j({\bf x})<\pi_{j+1}({\bf x})=\ldots=\pi_i({\bf x})= e<\pi_{i+1}({\bf x})$ then the following holds 
\begin{eqnarray*}
N(\pi_{i+1}({\bf x}))< e &=&N(\pi_{i}({\bf x}))=\ldots= N(\pi_{j+1}({\bf x}))< N(\pi_{j}({\bf x}))
\end{eqnarray*} 
with $j+1\leq i$ and $j\geq 0$. Hence, the equations as follows are verified:
\begin{eqnarray*}
& &\varphi^{\mathcal{N}}(\mathcal{N}(\textbf{x}))\\
 & = & (\psi(\pi_1(\mathcal{N}(\textbf{x}))), \ldots, \psi(\pi_{n-i}(\mathcal{N}(\textbf{x}))),\underbrace{e, \ldots, e,}_{i-(j+1)\ \textrm{times}}  N(\psi(N(\pi_{n-j+1}(\mathcal{N}(\textbf{x}))))),  \ldots, N(\psi(N(\pi_{n}(\mathcal{N}(\textbf{x})))))) \\& & \mbox{ by Eq.(\ref{NpsiN})} \\
 & = & (\psi(N(\pi_n(\textbf{x}))), \ldots, \psi(N(\pi_{i+1}(\textbf{x}))), \underbrace{e, \ldots, e,}_{i-(j+1)\  \textrm{times}}  N(\psi(N(N(\pi_{j}(\textbf{x}))))),\ldots,  N(\psi(N(N(\pi_1(\textbf{x}))))))\\& & \mbox{ since $N$ is strong} \\ 
& = & (N(N(\psi(N(\pi_n(\textbf{x}))))), \ldots, N(N(\psi(N(\pi_{i+1}(\textbf{x}))))),  \underbrace{e, \ldots, e,}_{i-(j+1)\ \textrm{times}} N(\psi(\pi_{j}(\textbf{x}))), \ldots, N(\psi(\pi_1(\textbf{x})))) \\& &\mbox{ since $N$ is strong} \\
& = & \widetilde{N}(\psi(\pi_1(\textbf{x}))), \ldots, \psi(\pi_{j}(\textbf{x})),\underbrace{\psi(e), \ldots, \psi(e),}_{i-(j+1)\ \textrm{times}} N(\psi(N(\pi_{i+1}(\textbf{x})))), \ldots, N(\psi(N(\pi_n(\textbf{x})))) \mbox{ by Eq.(\ref{tildeN})} \\
& = & \mathcal{N}(\psi(\pi_1(\textbf{x})), \ldots, \psi(\pi_{i}(\textbf{x})),  N(\psi(N(\pi_{i+1}(\textbf{x})))), \ldots, N(\psi(N(\pi_n(\textbf{x}))))) \\ 
& = & \mathcal{N}(\pi_1(\varphi({\bf x})), \ldots, \pi_{i}(\varphi({\bf x})), \pi_{i+1}(\mathcal{N}(\varphi(\mathcal{N}({\bf x})))),\ldots,  \pi_{n}(\mathcal{N}(\varphi(\mathcal{N}({\bf x}))))) \  \mbox{  by Eq.(\ref{NpsiN}),}  \\& & \mbox{based on results of Theorem~\ref{theoremvarphiN}} \\
& = & \mathcal{N}(\varphi^{\mathcal{N}}(\textbf{x})).
\end{eqnarray*}

Therefore, $\varphi^{\mathcal{N}}$ is $\mathcal{N}$-preserving $n$-dimensional automorphism. Now we will proof that all $\mathcal{N}$-preserving $n$-dimensional automorphisms have the form of Equation (\ref{varphiN}). Suppose that there exists an $\mathcal{N}$-preserving $n$-dimensional automorphism $\varphi': L_n([0,1])\rightarrow L_n([0,1])$. Then by Theorem~\ref{Npreserving}, $\psi': [0,1]\rightarrow [0,1]$ defined by $\psi'(x)=\pi_1(\varphi'(/x/))$ is an $N$-preserving automorphism. But, by Proposition \ref{bedregalprop2.6}, there exists an automorphism $\psi'': [0,e]\rightarrow [0,e]$ such that $\psi'=\psi''^{N}$. Let $\varphi''=\widetilde{\psi}''$. Hence, if $\textbf{x}\leq /e/$, then $\pi_i(\textbf{x})\leq e$ and so
\begin{eqnarray*}
\varphi'(\textbf{x}) & = & (\psi'(\pi_1(\textbf{x})), \ldots, \psi'(\pi_n(\textbf{x})))  \mbox{ based on results of  Theorem~\ref{Npreserving}}\\
                             & = & (\psi''^N(\pi_1(\textbf{x})), \ldots, \psi''^N(\pi_n(\textbf{x})))  \mbox{ based on results of Proposition~\ref{bedregalprop2.6}}\\
                             & = & (\psi''(\pi_1(\textbf{x})), \ldots, \psi''(\pi_n(\textbf{x})))  \mbox{ by Eq.(\ref{rhoN})}\\
                             & = & \widetilde{\psi}''(\pi_1(\textbf{x}), \ldots, \pi_n(\textbf{x})) \mbox{ by Eq.(\ref{widetildepsi})}\\
                             & = & \varphi''(\textbf{x}) \\
                             & = & \varphi''^{\mathcal{N}}(\textbf{x})  \mbox{ by Eq.(\ref{varphiN})}
\end{eqnarray*}
If $/e/<\textbf{x}$, then
\begin{eqnarray*}
\varphi'(\textbf{x}) & \hspace{-0.1cm}= \hspace{-0.1cm}&  (\psi'(\pi_1(\textbf{x})), \ldots, \psi'(\pi_n(\textbf{x})))  \mbox{ based on results of  Theorem~\ref{Npreserving}}\\
                             & \hspace{-0.1cm}= \hspace{-0.1cm} & (\psi''^N(\pi_1(\textbf{x})), \ldots, \psi''^N(\pi_n(\textbf{x})))   \mbox{ based on results of Proposition~\ref{bedregalprop2.6}}\\
                             & \hspace{-0.1cm}= \hspace{-0.1cm} & (N(\psi''(N(\pi_1(\textbf{x})))), \ldots, N(\psi''(N(\pi_n(\textbf{x}))))) \mbox{ by Eq.(\ref{preserving auto})}\\
                             & \hspace{-0.1cm}= \hspace{-0.1cm} & \widetilde{N}(\psi''(N(\pi_n(\textbf{x}))), \ldots, \psi''(N(\pi_1(\textbf{x}))))  \mbox{ by Eq.(\ref{tildeN})} \\
                             & \hspace{-0.1cm}= \hspace{-0.1cm} & \mathcal{N}(\psi''(N(\pi_n(\textbf{x}))), \ldots, \psi''(N(\pi_1(\textbf{x}))))  \mbox{ based on  Theorem~\ref{N=N}}\\
                             & \hspace{-0.1cm}= \hspace{-0.1cm} & \mathcal{N}(\widetilde{\psi}''(N(\pi_n(\textbf{x})), \ldots, N(\pi_1(\textbf{x}))))  \mbox{ by Eq.(\ref{widetildepsi})}\\
                             & \hspace{-0.1cm}= \hspace{-0.1cm} & \mathcal{N}(\varphi''(N(\pi_n(\textbf{x})), \ldots, N(\pi_1(\textbf{x}))))  \mbox{ based on results of Theorem~\ref{bedregal3.4}}\\                             
                             & \hspace{-0.1cm}= \hspace{-0.1cm} & \mathcal{N}(\varphi''(\widetilde{N}(\pi_1(\textbf{x})), \ldots, \pi_n(\textbf{x}))) \mbox{ by Eq.(\ref{tildeN})}\\
                            & \hspace{-0.1cm}= \hspace{-0.1cm} & \mathcal{N}(\varphi''(\mathcal{N}(\textbf{x})))) \mbox{ based on results of  Theorem~\ref{N=N}}\\
                             & \hspace{-0.1cm}= \hspace{-0.1cm} & \varphi''^{\mathcal{N}}(\textbf{x})  \mbox{ by Eq.(\ref{varphiN})}
\end{eqnarray*}
If $\pi_i({\bf x})\leq /e/<\pi_{i+1}({\bf x})$ then
\begin{eqnarray*}
\varphi'(\textbf{x})  \hspace{-0.1cm}  & \hspace{-0.1cm} = \hspace{-0.1cm} & \hspace{-0.1cm} (\psi'(\pi_1(\textbf{x})), \ldots, \psi'(\pi_n(\textbf{x})))  \mbox{ based on Theorem~\ref{Npreserving}}\\
                            \hspace{-0.1cm} & \hspace{-0.1cm} = \hspace{-0.1cm} & \hspace{-0.1cm} (\psi''^N(\pi_1(\textbf{x})), \ldots, \psi''^N(\pi_n(\textbf{x}))) \mbox{ based on results of Proposition~\ref{bedregalprop2.6}}\\
                             \hspace{-0.1cm} & \hspace{-0.1cm} = \hspace{-0.1cm} & \hspace{-0.1cm} (\psi''(\pi_1(\textbf{x})), \ldots, \psi''(\pi_{i}(\textbf{x})), N(\psi''(N(\pi_{i+1}(\textbf{x})))), \ldots, N(\psi''(N(\pi_n(\textbf{x})))))\\ 
                             \hspace{-0.1cm} & \hspace{-0.1cm} = \hspace{-0.1cm} & \hspace{-0.1cm} (\pi_1(\varphi''({\bf x})), \ldots, \pi_i(\varphi''({\bf x})), \pi_{i+1}(\mathcal{N}(\varphi''(\mathcal{N}({\bf x})))),\ldots,  \pi_{n}(\mathcal{N}(\varphi''(\mathcal{N}({\bf x})))))   \mbox{ by Eq.(\ref{NpsiN})} \\
& = & \varphi''^{\mathcal{N}}(\textbf{x}).
\end{eqnarray*}
Therefore, $\varphi'=\varphi''^{\mathcal{N}}$, i.e., all $\mathcal{N}$-preserving $n$-dimensional automorphisms have the form of Equation (\ref{varphiN}).
\end{proof}\\

The following result is analogous to Proposition \ref{bedregalprop2.7}.

\begin{proposition}
Let $\mathcal{N}$ be an $n$-dimensional strong fuzzy negation. Then, $(\varphi^{\mathcal{N}})^{-1}$ is an $\mathcal{N}$-preserving $n$-dimensional automorphism. 
\end{proposition}

\begin{proof}
By Theorem \ref{theoremvarphiN}, $\varphi^{\mathcal{N}}$ is an $\mathcal{N}$-preserving $n$-dimensional automorphism. Let $\textbf{x}\in L_n([0,1])$, then
\begin{eqnarray*}
(\varphi^{\mathcal{N}})^{-1}(\mathcal{N}(\textbf{x})) & = & (\varphi^{\mathcal{N}})^{-1}(\mathcal{N}(\varphi^{\mathcal{N}}((\varphi^{\mathcal{N}})^{-1}(\textbf{x}))))\\
                                                                                & = & (\varphi^{\mathcal{N}})^{-1}(\varphi^{\mathcal{N}}(\mathcal{N}((\varphi^{\mathcal{N}})^{-1}(\textbf{x}))))   \mbox{ by Eq.(\ref{Npreservingfor})}\\
                                                                                & = & \mathcal{N}((\varphi^{\mathcal{N}})^{-1}(\textbf{x}))
\end{eqnarray*} 
Therefore, by Eq.(\ref{Npreservingfor}), $(\varphi^{\mathcal{N}})^{-1}$ is also an $\mathcal{N}$-preserving $n$-dimensional automorphism.
\end{proof}

\section{Conclusion}
The principal research question considered in this paper is the following: how can the main properties of (strong) fuzzy negations on $L([0,1])$ be preserved by representable (strong) fuzzy negation on $L_n([0,1])$, mainly related to the analysis of degenerate elements and equilibrium points?  

Our aim was to design $n$-dimensional fuzzy negations and investigate one special extension from $L([0,1])$ -- the representable fuzzy negations on  $L_n([0,1])$, summarizing the class of such functions which are continuous and monotone by part. 

 In Theorem~\ref{subseq monotone}, $n$-representable interval negations on $L_n([0,1])$ were discussed by stating the necessary and sufficient conditions  which one can obtain an $n$-dimensional fuzzy negations from fuzzy negation $L([0,1])$. Further results also consider the subclass   of strong interval negations, with additional analysis of degenerate elements and $n$-dimensional equilibrium points in its $n$-membership functions. 
 
Theorem~\ref{theondsfn-nda} states the relationship between 
an $n$-dimensional strong fuzzy negation  $\mathcal{N}: L_n([0,1])\rightarrow L_n([0,1])$  and $n$-dimensional automorphism on $L_n([0,1])$.
The conjugate based on $n$-dimensional fuzzy negations  provides a method to obtain other $n$-dimensional fuzzy negations, in which properties of representable fuzzy negations on $L_n([0,1])$ are preserved. 

Extending the previous work~\cite{bedregal10,navara99} on theoretical research interval case  related to  $n$-dimensional version of fuzzy negation, Theorem~\ref{theoremvarphiN} provides  expression for all  $\mathcal{N}$-preserving $n$-dimensional automorphisms in $L_n([0,1])$ together with its reverse construction.

Further works investigate other fuzzy connectives, as implications and bi-implications, along with their representable classes, conjugate and dual constructions. 
In addition, we will investigate other orders for $L_n([0,1])$ such as admissible orders on n-dimensional intervals in the sense of \cite{BFK13,BGB13,dSBS16}.


%





\begin{thebibliography}{1}

  
\bibitem{alsina} C. Alsina, M.J. Frank, B. Scheweizer, \emph{Associative Funtions - Triangular Norms and Copulas}, World Scientific Publishing, Danvers, MA, 2006.

\bibitem{atanassov99} K. Atanassov, \emph{Intuitionistic Fuzzy Sets, Theory and Applications}, Physica-Verlag, Heidelberg, 1999.

\bibitem{baets} B. De Baets, R. Mesiar, \emph{Triangular Norms on Product Lattices}, Fuzzy Sets and Systems, 104 (1999) 61-75.

\bibitem{helida} B. Bedregal, H. Santos, R.C. Bedregal \emph{T-norms on Bouded lattices: t-norm morphisms and operators}.  Fuzzy Systems, 2006 IEEE International Conference on Fuzzy Systems (2006) 22 -- 28. DOI: 10.1109/FUZZY.2006.1681689

\bibitem{bedregal10} B.C. Bedregal, \emph{On interval fuzzy negations}, Fuzzy Sets and Systems 161 (17) (2010) 2290--2313. 

\bibitem{IPMU12} B. Bedregal, G. Beliakov, H. Bustince, J. Fernandez, A. Pradera,  R. Reiser, \emph{Negations Generated by Bounded Lattices t-Norms}, in: Proceedings of IPMU 2012, part III, CCIS 229, 326--335, Springer-Verlag.

\bibitem{bedregal11} B. Bedregal, G. Beliakov, H. Bustince, T. Calvo, J. Fern\'andez, R. Mesiar, D. Paternain, \emph{A characterization theorem for t-representable 
n-dimensional triangular norms}, in: Proceeding of Eurofuse 2011, Vol. 107 of Advances in Intelligent and Soft Computing (2011) 103--112. 

\bibitem{bedregal12} B. Bedregal, G. Beliakov, H. Bustince, T. Calvo, R. Mesiar, D. Paternain, \emph{A class of fuzzy multisets with a fixed number of memberships}, 
Information Sciences 189 (2012) 1--17. 

\bibitem{beliakov} G. Beliakov, A. Pradera, T. Calvo, \emph{Aggregation functions: A Guide for Practitioners}, Springer, Berlin, 2007.

\bibitem{bustince15} H. Bustince, E. Barrenechea, M. Pagola, J. Fernandez, Z. Xu, B. Bedregal, J. Montero, H. Hagras, F. Herrera, B. De Baets, 
\emph{A historical account of types of fuzzy sets and their relationships}, IEEE Transactions on Fuzzy Systems 4 (1) (2016)  179--194

\bibitem{bustince03} H. Bustince, P. Burillo, F. Soria, \emph{Automorphisms, negations and implication operators}, Fuzzy Sets and Systems 134 (2003) 209--229.


\bibitem{BFK13} H. Bustince, J. Fernandez, A. Koles\'arov\'a, R. Mesiar,
\emph{Generation of linear orders for intervals by means of aggregation functions}, Fuzzy Sets and Systems 220 (2013) 69--77.

\bibitem{BGB13} H. Bustince, M. Galar, B. Bedregal, A. Koles\'arov\'a, R. Mesiar,
\emph{A new approach to interval-valued Choquet integrals and the problem of ordering in interval-valued fuzzy Set applications},
IEEE Trans. Fuzzy Systems 21(6) (2013) 1150--1162.

\bibitem{bustince99} H. Bustince, J. Montero, M. Pagola, E. Barrenechea, D. Gomes, \emph{A survey of interval-valued fuzzy sets}, in: W. Predrycz, A. Skowron, 
V. Kreinovich (Eds.), Handbook of Granular Computing, John Wiley \& Sons Ltd., West Sussex, 2008, 491--515, Chapter 22.


\bibitem{costa} C.G. Da Costa, B.C. Bedregal, A.D. Doria Neto, \emph{Relating De Morgan triples with Atanassov's intuitionistic De Morgan triples via automorphisms}, International Journal of Approximate Reasoning 52 (2011) 473--487.

\bibitem{dSBS16} I.A. Da Silva, B. Bedregal, R.H.N. Santiago, \emph{On admissible total orders for interval-valued intuitionistic fuzzy membership degrees}, 
Fuzzy Information and Engineering 8(2) (2016) 169--182.


\bibitem{FodorRoubens1994} J. Fodor, M. Roubens, \emph{Fuzzy Preference Modelling and Multicriteria Decision Support}, Kluwer Academic Publisher, Dordrecht, 1994.

\bibitem{goguen} J.A. Goguen \emph{L-fuzzy sets}, Journal of Mathematical Analysis and Applications, 18  (1) (1967) 145--174.

\bibitem{klement} 	E.P. Klement, R. Mesiar, E. Pap, \emph{Triangular Norms}, Kluwer Academic Publishers, Dordrecht, 2000.

\bibitem{klir95} G.J. Klir, B. Yuan, \emph{Fuzzy Sets and Fuzzy Logics: Theory and Applications}. Prentice Hall PTR, Upper Saddle River, NJ, 1995.

\bibitem{mezzomo16} I. Mezzomo, B.C. Bedregal, R.H.S. Reiser, H. Bustince, D. Partenain,  \emph{On n-dimensional strict fuzzy negations}, 2016 IEEE International Conference on Fuzzy Systems (FUZZ-IEEE), (2016) 301 -- 307. DOI: 10.1109/FUZZ-IEEE.2016.7737701

\bibitem{navara99} M. Navara, \emph{Characterization of measures based on strict triangular norms}, Journal of Mathematical Analysis and Applications, 236 (2) (1999) 370--383.

\bibitem{palmeira} E.S. Palmeira, B.C. Bedregal, \emph{Extension of n-dimensional lattice-valued negations}. Decision Making and Soft Computing: Proceedings of the 11th International  FLINS Conference (2014) 306 -- 311.

\bibitem{palmeira14} E.S. Palmeira, B.C. Bedregal, R. Mesiar, J. Fernandez, \textit{A new way to extend t-norms, t-conorms and negations}, Fuzzy Set and Systems 240 (2014) 1-21.

\bibitem{shang} Y. Shang, X. Yuan, E.S. Lee, \emph{The n-dimensional fuzzy sets and Zadeh fuzzy sets based on the finite valued fuzzy sets}, Computers \& Mathematics 
with Applications 60 (2010) 442--463.

\bibitem{trillas79} E. Trillas, \emph{Sobre funciones de negaci\'on en la teoria de conjuntos difusos}, Stochastica, 3 (1979) 47--59.


\end{thebibliography}
%

\end{document}